\documentclass[10pt,a4paper]{article}

\usepackage{fullpage}
\usepackage{microtype}
\usepackage{libertine}

\usepackage{enumerate}
\usepackage{hyperref}
\usepackage{amsmath,amssymb,amsthm}
\usepackage{tikz}
\usepackage{algorithm}
\usepackage{algpseudocode}
\usepackage{fullpage}
\usepackage{todonotes}
\usepackage{authblk}
\usepackage{xspace}
\usepackage[normalem]{ulem}

\usetikzlibrary{calc}

\usepackage[absolute]{textpos}
\usepackage[all=normal,bibliography=tight]{savetrees}

\newtheorem{theorem}{Theorem}[section]
\newtheorem{lemma}[theorem]{Lemma}
\newtheorem{corollary}[theorem]{Corollary}

\newtheorem{conjecture}[theorem]{Conjecture}
\newtheorem{claim}{Claim}

\theoremstyle{definition}

\theoremstyle{remark}

\newcommand{\Oh}{\mathcal{O}}

\newcommand{\inst}{\mathcal{I}}
\newcommand{\flow}{\mathcal{P}}
\newcommand{\witnessflow}{\widehat{\flow}}
\newcommand{\resreach}[1]{\mathrm{RReach}(#1)}
\newcommand{\lastreach}[2]{\mathrm{LastReach}(#1,P_{#2})}
\newcommand{\leader}[3]{\mathrm{leader}_{#1}(#2,#3)}
\newcommand{\closure}[1]{\mathrm{cl}(#1)}
\newcommand{\block}{\mathbf{B}}
\newcommand{\longbackarcs}{\mathbf{L}}
\newcommand{\ellthreshold}{\ell^{\mathrm{big}}}
\newcommand{\leftblock}{b_\leftarrow}
\newcommand{\rightblock}{b_\rightarrow}
\newcommand{\corecut}[1]{\mathrm{core}(#1)}
\newcommand{\corecutG}[2]{\mathrm{core}_{#2}(#1)}
\newcommand{\cutlb}{\kappa}

\newcommand{\bundles}{\mathcal{B}}
\newcommand{\weight}{\omega}

\newcommand{\grantthankstext}{This research is a part of a project that have received funding from the European Research Council (ERC) under the European Union's Horizon 2020 research and innovation programme
Grant Agreement 714704 (M. Pilipczuk). Eun Jung Kim is supported by the grant from French National Research Agency under JCJC program (ASSK: ANR-18-CE40-0025-01).}
\newcommand{\grantthanks}{\parbox[t]{0.78\linewidth}{\grantthankstext{}}\ \parbox[t]{0.18\linewidth}{~\\[-3mm]\includegraphics[height=30px]{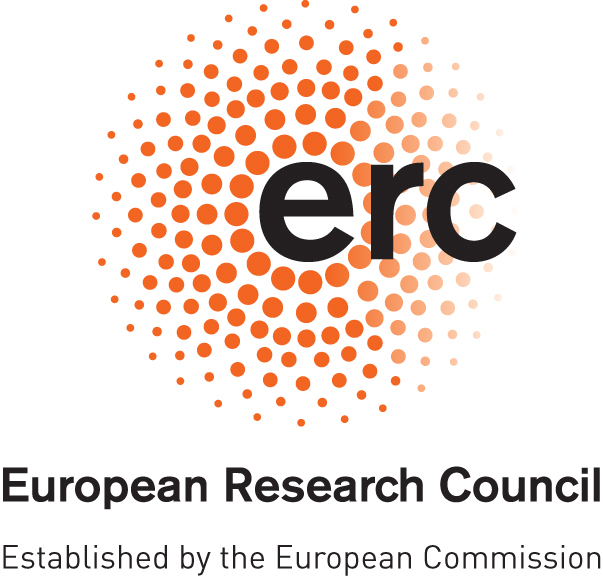}\ \includegraphics[height=30px]{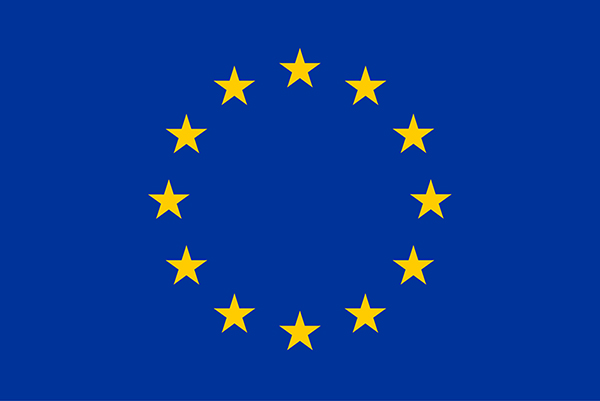}}}

\title{Flow-augmentation I: Directed graphs%
\thanks{\grantthanks{}\\A preliminary version of this work was presented at STOC 2022~\cite{dfl-stoc}.}}
\date{}

\author[1]{Eun Jung Kim}
\author[2]{Stefan Kratsch}
\author[3]{Marcin Pilipczuk}
\author[4]{Magnus Wahlstr\"{o}m}
\affil[1]{Universit\'{e} Paris-Dauphine, PSL Research University, CNRS, UMR 7243, LAMSADE, 75016, Paris, France.}
\affil[2]{Humboldt-Universit\"at zu Berlin, Germany}
\affil[3]{University of Warsaw, Warsaw, Poland}
\affil[4]{Royal Holloway, University of London, TW20 0EX, UK}

\begin{document}

\begin{titlepage}
\def\thepage{}
\thispagestyle{empty}

\maketitle

\begin{abstract}
We show a flow-augmentation algorithm in directed graphs:
There exists a randomized polynomial-time algorithm that,
given a directed graph $G$, two vertices $s,t \in V(G)$, and an integer $k$, 
adds (randomly) to $G$ a number of arcs such that for every minimal $st$-cut $Z$ in $G$
of size at most $k$, with probability $2^{-\mathrm{poly}(k)}$ the set $Z$ becomes a
\emph{minimum} $st$-cut in the resulting graph. 
We also provide a deterministic counterpart of this procedure.

The directed flow-augmentation tool allows us to prove fixed-parameter tractability
of a number of problems parameterized by the cardinality of the deletion set,
whose parameterized complexity status was repeatedly posed as open problems:
\begin{enumerate}
\item \textsc{Chain SAT}, defined by Chitnis, Egri, and Marx [ESA'13, Algorithmica'17], %
\item a number of weighted variants of classic directed cut problems, such as \textsc{Weighted $st$-Cut} or \textsc{Weighted Directed Feedback Vertex Set}.
\end{enumerate}
By proving that \textsc{Chain SAT} is FPT, we confirm a conjecture of Chitnis, Egri, and Marx
that, for any graph $H$, if the \textsc{List $H$-Coloring} problem is polynomial-time solvable,
then the corresponding vertex-deletion problem is fixed-parameter tractable.

\end{abstract}

\end{titlepage}

\section{Introduction}
The study of graph separation problems has been one of the more vivid areas of parameterized
complexity in the recent 10--15 years. 
The term ``graph separation problems'' is here used widely, and captures a number of classic graph problems where,
given a (undirected or directed) graph $G$ with a cut budget $k$
(and possibly some annotations, such as terminal vertices), one aims at obtaining some
separation via at most $k$ edge or vertex deletions. For example, 
the classic \textsc{$st$-Cut} problem asks to delete at most $k$ edges so that 
there is no $s$-to-$t$ path in the resulting graph
and the \textsc{Feedback Vertex Set} asks to remove at most $k$ vertices
so that the resulting graph does not contain any cycles (i.e., is a forest 
in the undirected setting or a DAG in the directed setting). 
In all these problems, the cardinality of the deletion set is a natural parameter to study. 

In 2004, Marx introduced the notion of \emph{important separators}~\cite{Marx04,Marx06}
that turned out to be the key to fixed-parameter tractability of \textsc{Multiway Cut}
and \textsc{Directed Feedback Vertex Set}~\cite{ChenLLOR08}, among many other examples. 
In subsequent years, the study of graph separation problems resulted in a rich toolbox
of algorithmic techniques, such as \emph{shadow removal}~\cite{MarxR14}, 
treewidth reduction~\cite{MarxOR13}, or randomized contractions~\cite{ChitnisCHPP16,CyganLPPS19,CyganKLPPSW21}.

This resulted in a relatively mature, but not fully complete, 
landscape of the parameterized complexity of graph separation problems in undirected graphs.
The remaining questions tackled mostly weighted variants of the problems 
(e.g., the weighted \textsc{Multicut} problem) or more intricate variants
of classic problems (e.g., the \textsc{Coupled Mincut} problem\footnote{
In this problem, we are given an undirected graph $G$ with distiguished terminals $s,t \in E(G)$
and some edges coupled up in pairs. The question is to separate $s$ from $t$ at minimum
cost under the following conditions: at cost $1$ one can delete an unpaired edge or both
edges from a pair, but for every undeleted pair $(e,e')$, one cannot leave
both $e$ and $e'$ reachable from $s$ in the remaining graph.}).

For directed graphs, the chartered landscape is much less complete. 
The notion of important separators and related technique of shadow removal
generalizes to directed graphs, leading to fixed-parameter tractability of
\textsc{Directed Feedback Vertex Set}~\cite{ChenLLOR08},
\textsc{Directed Multiway Cut}~\cite{ChitnisHM13}, and \textsc{Directed Subset Feedback Vertex Set}~\cite{ChitnisCHM15}.
A number of problems whose undirected counterparts are FPT turned out to be $W[1]$-hard in
the directed setting, including \textsc{Directed Multicut}~\cite{ChitnisHM13,PilipczukW18a} or
\textsc{Directed Odd Cycle Transversal}~\cite{LokshtanovR0Z20}. 

However, these results are far from satisfactorily charting the parameterized complexity of
directed graph separation problems. 
Arguably, we seem to lack algorithmic techniques. Most notably, the powerful
treewidth reduction theorem~\cite{MarxOR13}, stating that (in undirected graphs) all separations
of size at most $k$ between two fixed terminals live in a part of the graph with treewidth
bounded by $2^{\Oh(k)}$, seems not to have any meaningful counterpart in directed graphs. 
As a result, essentially all known FPT algorithms 
for graph separation problems in directed graphs rely in some part on important separators,
which is a greedy argument bounding the number of cuts between two terminals of bounded size that have
inclusion-wise maximal set of vertices reachable from one of the terminals. 
For problems where such a ``greedy'' aspect of the solution cannot be assumed,
such as problems with weights or annotations, this method fails to apply.

On the other hand, despite efforts in the last years, we were not able to prove lower bounds
for FPT algorithms for many directed graph separation problems. 
This suggests that maybe there are still more algorithmic techniques to be explored,
leading to more positive tractability results. 

In this work, we provide such a technique, which we call \emph{flow-augmentation}.
\begin{theorem}\label{thm:intro}
There exists a randomized polynomial-time algorithm that, 
given a directed graph $G$, two vertices $s,t \in V(G)$, and an integer $k$,
outputs a set $A \subseteq V(G) \times V(G)$ such that the following holds:
for every minimal $st$-cut $Z \subseteq E(G)$ of size at most $k$,
with probability $2^{-\Oh(k^4 \log k)}$
$Z$ remains an $st$-cut in $G+A$ and, furthermore, $Z$ is a \emph{minimum} $st$-cut in $G+A$.
\end{theorem}
Here, a set $Z \subseteq E(G)$ is an \emph{$st$-cut} if there is no path from $s$ to $t$ in $G-Z$;
it is \emph{minimal} if no proper subset of $Z$ is an $st$-cut and \emph{minimum}
if it is of minimum possible cardinality. 
By $G+A$ we mean the graph $G$ with all elements of $A$ added as infinity-capacity arcs. 

The proof of Theorem~\ref{thm:intro} is presented in Section~\ref{sec:flow-augmentation}. 
(There, we actually prove a slight generalization that is handy in some applications.)
We also provide a deterministic counterpart.
\begin{theorem}\label{thm:intro-det}
There exists an algorithm that,
given a directed graph $G$, two vertices $s,t \in V(G)$, and an integer $k$,
in time $2^{\Oh(k^4 \log k)} |V(G)|^{\Oh(1)}$ 
outputs a set $\mathcal{A} \subseteq 2^{V(G) \times V(G)}$ of size $2^{\Oh(k^4 \log k)} (\log n)^{\Oh(k^3)}$
such that for every minimal $st$-cut $Z \subseteq E(G)$ of size at most $k$
there exists $A \in \mathcal{A}$
such that $Z$ remains an $st$-cut in $G+A$ and, furthermore, $Z$ is a \emph{minimum} $st$-cut in $G+A$.
\end{theorem}
We remark that
\begin{equation}\label{eq:intro-det}
 (\log n)^{\Oh(k^3)} = 2^{\Oh(k^3 \log \log n)} \leq 2^{\Oh(k^4 + (\log \log n)^4)} = 2^{\Oh(k^4)} n^{o(1)}. 
\end{equation}
However, in some applications in~\cite{csp-arxiv}, we stack a bounded-in-$k$ number
of usages of Theorem~\ref{thm:intro-det} on top of each other, and for the sake of the analysis there
it is more
convenient to keep the part of the running time bound that depends on $n$
in the form of $(\log n)^{f(k)}$, as opposed to $f(k) \cdot n^{o(1)}$ in~\eqref{eq:intro-det}.

\paragraph{Applications}
To illustrate the applicability of the directed flow-augmentation, let us first
consider the \textsc{Weighted $st$-Cut} problem. Here, we are given
a directed graph $G$ with two terminals $s,t\in V(G)$, a weight function
$\weight:E(G) \to \mathbb{Z}_+$, and two integers $k,W \in \mathbb{Z}_+$,
and we ask for an $st$-cut $Z$ of cardinality at most $k$ and total weight at most $W$.
This problem is known to be NP-hard and FPT when parameterized by $k+W$~\cite{KratschLMPW20}.

By using directed flow-augmentation, we can ensure that the sought solution $Z$
is actually a \emph{minimum} $st$-cut (i.e., of minimum cardinality).
Then, a solution can easily be found in polynomial-time:
take $M := 1 + \sum_{e \in E(G)} \weight(e)$, set the capacity of every 
edge $e$ as $\weight(e) + M$, and find an $st$-cut of minimum capacity. 
This yields the following.
\begin{theorem}
\textsc{Weighted $st$-Cut} can be solved in time $2^{\Oh(k^4 \log k)} n^{\Oh(1)}$.
\end{theorem}

The above approach turns out to apply more generally.
An instance of \textsc{Bundled Cut} consists of a directed multigraph $G$, vertices $s,t \in V(G)$, a nonnegative integer $k$, and a family $\bundles$ of pairwise disjoint subsets of $E(G)$.
The elements of $\bundles$ are henceforth called \emph{bundles}.
A \emph{cut} in a \textsc{Bundled Cut} instance $\inst = (G,s,t,k,\bundles)$ is an $st$-cut $Z$ 
with $Z \subseteq \bigcup \bundles$. 
The \textsc{Bundled Cut} problem asks for an $st$-cut that intersects at most $k$ bundles, that is,
    $|\{B \in \bundles~|~Z \cap B \neq \emptyset\}| \leq k$. 
One can also define a weighted variant of the \textsc{Bundled Cut} problem where
every bundle $B \in \bundles$ is equipped with a weight $\weight(B) \in \mathbb{Z}_+$,
we are given also a weight bound $W$, 
and we ask for an $st$-cut $Z$ that intersects at most $k$ bundles whose total weight
is at most $W$.

If all bundles are singletons, then \textsc{Bundled Cut} just asks for an $st$-cut of size at most
$k$ (with the edges of $E(G) \setminus \bigcup \bundles$ being undeletable), so it is
polynomial-time solvable. It is known that, when restricted to bundles of size $2$,
\textsc{Bundled Cut} is $W[1]$-hard when parameterized by $k$~\cite{MarxR09}. 
To get tractability, 
we define the following restriction.
An edge $e \in E(G)$ is \emph{soft} if $e \in \bigcup \bundles$ and \emph{crisp} otherwise.
An edge $e \in E(G)$ is \emph{deletable} if it is soft and there is no parallel arc to $e$
that is crisp, and \emph{undeletable} otherwise.
Clearly, we can restrict ourselves to cuts $Z$ that consist of deletable arcs only. 
An instance $(G,s,t,k,\bundles)$ has \emph{pairwise linked deletable edges} if for every
$B \in \bundles$ and every two deletable edges $e_1,e_2 \in B$, there exists a path in $G$
from an endpoint of $e_1$ to an endpoint of $e_2$ that uses only edges of $B$
and undeletable edges. Note that this path may have length zero, 
i.e., it suffices that $e_1$ and $e_2$ intersect. 

Using directed flow-augmentation, we show the following.
\begin{theorem}\label{thm:intro:bundled-cut}
\textsc{Weighted Bundled Cut}, restricted to instances with pairwise linked deletable edges,
 is FPT when parameterized by $k$ and the maximum number of deletable edges in a bundle. 
\end{theorem}

For an integer $\ell \geq 1$, the \textsc{$\ell$-Chain SAT} problem is 
the \textsc{Bundled Cut} problem, where every bundle is a path of length at most $\ell$. 
At ESA'13, Chitnis, Egri, and Marx~\cite{ChitnisEM13,ChitnisEM17} defined the \textsc{$\ell$-Chain SAT} problem 
and showed that fixed-parameter tractability of 
\textsc{$\ell$-Chain SAT} (for every fixed $\ell \geq 1$) is equivalent to the following 
conjecture.
\begin{conjecture}[Conjecture 1.1 of~\cite{ChitnisEM17}]\label{conj:H-col}
For every graph $H$, if \textsc{List $H$-Coloring} problem is polynomial-time solvable,
then the vertex-deletion variant (delete at most $k$ vertices from the input graph
    to obtain a yes-instance to \textsc{List $H$-Coloring}) is fixed-parameter tractable
when parameterized by $k$. 
\end{conjecture}

Clearly, an \textsc{$\ell$-Chain SAT} instance is a \textsc{Bundled Cut} instance
with at most $\ell$ deletable edges in a bundle and pairwise linked deletable edges
(the bundle itself provides the desired connectivity between deletable edges in a bundle). 
Hence, Theorem~\ref{thm:intro:bundled-cut} implies the following.
\begin{corollary}
\textsc{$\ell$-Chain SAT} is FPT when parameterized by $\ell$ and $k$, 
  even in the weighted setting. 
Consequently, Conjecture~\ref{conj:H-col} is confirmed.
\end{corollary}

By standard reductions (spelled out in Section~\ref{ss:wdfvs}),
a \textsc{Directed Feedback Vertex Set}
instance can also be represented as a \textsc{Bundled Cut} instance with pairwise linked
deletable edges
(with maximum size of a bundle and budget bounded linearly in the parameter of the input instance).
Furthermore, in case of a weighted variant%
\footnote{In \textsc{Weighted Directed Feedback Vertex Set}, 
  the input graph is equipped with vertex weights being positive integers
    and one asks for a solution of cardinality at most $k$
    and total weight bounded by a threshold given on input.}
the reduction preserves weights. 
Hence, we obtain the following.
\begin{corollary}\label{cor:wdfvs}
\textsc{Weighted Directed Feedback Vertex Set}, parameterized by the cardinality
of the deletion set, is FPT.
\end{corollary}
The question of 
parameterized complexity of \textsc{Weighted Directed Feedback Vertex Set} 
(which we answer affirmatively in Corollary~\ref{cor:wdfvs}) has been
asked e.g. at Recent Advances in Parameterized Complexity school in December 2017~\cite{rapcslides}
and in~\cite{LokshtanovRS16}.

\paragraph{Subsequent and accompanying work}
This paper is the first part in a series that explores variants and 
applicability of the new technique.
In this part, the main emphasis is on Theorem~\ref{thm:intro} itself, with aforementioned corollaries
stated as motivation. 

In the second part~\cite{ufl-soda}, we show that in undirected graphs one can get
an improved success probability of $2^{-\Oh(k \log k)}$ and a linear
running time for the augmentation procedure.
The proof is also arguably simpler. 

In the \textsc{Min SAT$(\Gamma)$} problem, we are given 
an instance of a Constraint Satisfaction Problem over language $\Gamma$ and an integer $k$
to delete at most $k$ clauses to get a satisfiable instance. 
Here, we restrict only to binary alphabet. 
Note that if $\Gamma$ allows unary clauses and equalities, \textsc{Min SAT$(\Gamma)$}
becomes the question of minimum (undirected) $st$-cut
and when $\Gamma$ allows unary clauses and implications, \textsc{Min SAT$(\Gamma)$}
becomes the question of minimum (directed) $st$-cut.
In the third part of the series~\cite{csp-arxiv}, we show full dichotomy for parameterized complexity
of \textsc{Min SAT$(\Gamma)$} for finite boolean languages $\Gamma$,
parameterized by the deletion budget $k$.
Here, flow-augmentation turned out to be pivotal for two new
isles of tractability, one of them being a wide generalization of the aforementioned \textsc{Coupled Mincut} problem.

In~\cite{weighted-arxiv}, Sharma and a subset of the current authors explored further
application in the realm of weighted graph separation problems. Using the new isles of tractability
of~\cite{csp-arxiv} as starting points, they showed tractability of weighted versions
of \textsc{Multicut} in undirected graphs 
and \textsc{Directed Subset Feedback Vertex Set}, among others.
Prior to~\cite{weighted-arxiv}, Galby et al~\cite{GalbyMSST22} used flow-augmentation to show
tractability of weighted \textsc{Multicut} in trees.
Flow-augmentation has been also used in~\cite{dir3mc-arxiv} to show fixed-parameter tractability
of \textsc{Multicut} in directed graphs with three terminal pairs, resolving another long-standing
open problem.

\section{Preliminaries}
\subsection{Cuts, flows}

All our graphs allow parallel edges. Edges may have capacities $1$ or $+\infty$. 
Since we never consider flows of value greater than $k$, a $+\infty$-capacity edge is equivalent to $(k+1)$ copies of an edge of capacity $1$. 
If $G$ is a graph and $A \subseteq V(G) \times V(G)$, then $G+A$ is the graph $G$ with every (ordered) pair of $A$ added as an arc \emph{with capacity $+\infty$}. 

For a directed graph $G$ and a set $X \subseteq V(G)$, the set $\delta_G^+(X)$ is the set of arcs with tails in $X$ and heads in $V(G) \setminus X$.
Similarly $\delta_G^-(X)$ is the set of arcs with tails in $V(G)\setminus X$ and heads in $X$. 
We write $\delta_G^+(v)$ and $\delta_G^-(v)$ for $\delta_G^+(\{v\})$ and $\delta_G^-(\{v\})$. When the graph $G$ under consideration is clear 
from the context, we omit the subscript $G$.

Let $G$ be a directed graph and let $s,t \in V(G)$.
An \emph{$st$-flow} is a collection $\flow$ of paths from $s$ to $t$ such that no edge of capacity $1$ lies on more than one path of $\flow$. 
The \emph{value} of the flow is the number of paths. 
By $\lambda_G(s,t)$ we denote the maximum possible value of an $st$-flow; note that it may happen that $\lambda_G(s,t) = +\infty$.
An $st$-flow is a \emph{maximum $st$-flow} or \emph{$st$-maxflow} if its value is $\lambda_G(s,t)$. 
By convention, if $\lambda_G(s,t) = +\infty$, then any $st$-flow that contains a path with all edges of capacity $+\infty$ is considered an $st$-maxflow.

A set $Z \subseteq E(G)$ is an \emph{$st$-cut} if it contains no edge of capacity $+\infty$ and there is no path from $s$ to $t$ in $G-Z$. 
An $st$-cut $Z$ is \emph{minimal} if no proper subset of $Z$ is an $st$-cut and
\emph{minimum} (or \emph{$st$-mincut}) if it has minimum possible cardinality.
By Menger's theorem, if $\lambda_G(s,t) < +\infty$ then the size of every $st$-mincut is exactly $\lambda_G(s,t)$
and there are no $st$-cuts if $\lambda_G(s,t) = +\infty$. 

An $st$-cut $Z$ is a \emph{star $st$-cut} if for every $(u,v) \in Z$, in the graph $G-Z$ there is a path from $s$ to $u$ but there is no path from $s$ to $v$.
Note that every minimal $st$-cut is a star $st$-cut, but the implication in the other direction does not hold in general. 
For a star $st$-cut $Z$ in $G$, by $\corecutG{Z}{G} \subseteq Z$ we denote the set of arcs $(u,v) \in Z$ such that there exists a path from $v$ to $t$ in $G-Z$. 
We drop the subscript if the graph $G$ is clear from the context. 

We observe the following.
\begin{lemma}
If $Z$ is a star $st$-cut in a graph $G$, then $\corecut{Z}$ is a minimal $st$-cut.
\end{lemma}

For an $st$-cut $Z$, the \emph{$s$-side} of $Z$ is the set of vertices reachable from $s$ in $G-Z$, and the \emph{$t$-side} of $Z$ is the complement of the $s$-side. 
Note that this is not symmetric, that is, we do not mandate that $t$ is reachable from all elements of the $t$-side in $G-Z$. 
In a star $st$-cut $Z$, all tails of edges of $Z$ are in the $s$-side of $Z$ and all heads of edges of $Z$ are in the $t$-side of $Z$.

We say that a set of arcs $A \subseteq V(G) \times V(G)$ is \emph{compatible} with a star $st$-cut $Z$ if the following holds:
for every $v \in V(G)$, there is a path from $s$ to $v$ in $G-Z$ if and only if there is a path from $s$ to $v$ in $(G+A)-Z$. 
Equivalently: the $s$-sides and $t$-sides of $Z$ are unchanged in $G$ and $G+A$, or put another way,  no arc of $A$ simultaneously has its tail in the $s$-side of $Z$ in $G$
and its head in the $t$-side of $Z$ in $G$. 

An immediate yet important observation the following.
\begin{lemma}\label{lem:star-cut-stays}
If $Z$ is a star $st$-cut in $G$ and $A \subseteq V(G) \times V(G)$ is compatible with $Z$, then
$Z$ is a star $st$-cut in $G+A$ as well. 
\end{lemma}
We remark that albeit in the setting of Lemma~\ref{lem:star-cut-stays} the cut $Z$ remains a star $st$-cut in $G+A$, 
it may happen that $\corecutG{Z}{G} \subsetneq \corecutG{Z}{G+A}$, as the arcs of $A$ may add some new reachability towards $t$. 

Assume $Z$ is a star $st$-cut in $G$ such that $\corecut{Z}$ is an $st$-mincut. 
An $st$-maxflow $\flow$ is a \emph{witnessing flow} if $E(\flow) \cap Z = \corecut{Z}$, that is, $\flow$ contains one edge of $\corecut{Z}$ on each flow path and no other edge of $Z$.
A witnessing flow may not exist in general, even if $\corecut{Z}$ is an $st$-mincut. 
However, our flow-augmentation procedure will ensure that not only $\corecut{Z}$ becomes an $st$-mincut in the augmented graph, but also a flow is returned that is a
witnessing flow in the augmented graph. 
Formally, for a star $st$-cut $Z$ in $G$, $A \subseteq V(G) \times V(G)$, and an $st$-maxflow $\witnessflow$ in $G+A$, we say that $(A,\witnessflow)$ is \emph{compatible} with $Z$
if $A$ is compatible with $Z$, $\corecutG{Z}{G+A}$ is an $st$-mincut in $G+A$, and $\witnessflow$ is a witnessing flow for $Z$ in $G+A$. 
Our flow-augmenting procedure will return a pair $(A,\witnessflow)$ that is compatible with a fixed star $st$-cut $Z$ with good probability. 

An edge $e \in E(G)$ is a \emph{bottleneck} edge if there exists an $st$-mincut that contains $e$.

For an $st$-flow $\flow$, the \emph{residual network} $G^\flow$ is constructed from $G$ by
first adding a capacity-$1$ arc $(u,v)$ for every $P \in \flow$ and $(v,u) \in E(P)$ (an edge $(u,v)$ is added multiple times if $(v,u)$ appears on multiple paths of $\flow$)
and then deleting all capacity-$1$ arcs on paths of $\flow$. 
An \emph{augmenting path} is a path from $s$ to $t$ in $G^\flow$. 
Recall that $\flow$ is an $st$-maxflow if and only if there is no augmenting path  for $\flow$ and $G$. 

Let $C$ be an $st$-mincut and $\flow$ be an $st$-maxflow. 
First, note that for every $P \in \flow$ there is exactly one edge of $C \cap E(P)$. 
This edge splits $P$ into a part in the $s$-side of $C$ and a part in the $t$-side of $C$.
We say that an edge or a vertex on $P$ is \emph{before} (\emph{after}) $C$ on $P$ if it is in the $s$-side ($t$-side, respectively) of $C$. 
We infer that every path in $G$ that goes from a vertex in the $s$-side of $C$ to a vertex in the $t$-side of $C$
visits an edge of $C$ and, furthermore, no augmenting path goes from a vertex in the $s$-side of $C$ to a vertex in the $t$-side of $C$.

By submodularity, if $Z_1$ and $Z_2$ are $st$-mincuts and $X_1$ is the $s$-side of $Z_1$ and $X_2$ is the $s$-side of $Z_2$, 
   then both $\delta^+(X_1 \cup X_2)$ and $\delta^+(X_1 \cap X_2)$ are $st$-mincuts.
This implies that there exists a unique $st$-mincut with inclusion-wise minimal $s$-side and a unique $st$-mincut with inclusion-wise maximal $s$-side.
We call them the $st$-mincut \emph{closest to $s$} and the $st$-mincut \emph{closest to $t$}.

\subsection{Color-coding and its derandomization}\label{ss:color-coding}

A standard color-coding step is the following randomized step: given two sets $A,B$,
we randomly sample a function $f: A \to B$. The goal is that for some unknown subset $A' \subseteq A$
and an unknown function $f' : A' \to B$,
we aim at $f$ extending $f'$. Obviously, this happens with probability $|B|^{-|A|'}$. 
In applications, usually $A'$ and $f'$ are some objects associated with the unknown sought solution.
For the deterministic version of flow-augmentation, we will need the following derandomization
statement (cf.~\cite{the-book} for a wider discussion).

\begin{theorem}\label{thm:color-coding}
Given two finite sets $A$ and $B$ and an integer $k$, one can in 
$|A|^{\Oh(1)} \cdot 2^{\Oh(k \log |B|)}$ time output a family $\mathcal{F}$
of functions $A \to B$ of size $2^{\Oh(k \log |B|)} \cdot \Oh(\log |A|)$ such that
for every $A' \subseteq A$ of size at most $k$ and every $f': A' \to B$, there
exists $f \in \mathcal{F}$ that extends $f'$.
\end{theorem}

We will also need the following random separation version; this exact statement
is taken from~\cite{ChitnisCHPP16,CyganKLPPSW21}.

\begin{theorem}\label{thm:rand-separation}
Given a set $U$ of size $n$ and integers $0 \leq a,b \leq n$, one
can in time $2^{\Oh(\min(a,b) \log (a+b))} n \log n$
construct a family $\mathcal{F}$ of at most $2^{\Oh(\min(a,b) \log (a+b))} \log n$
subsets of $U$ such that the following holds:
for any sets $A,B \subseteq U$ with $A \cap B = \emptyset$, $|A|\leq a$, and $|B|\leq b$,
there exists a set $S \in \mathcal{F}$ with $A \subseteq S$ and $B \cap S = \emptyset$.
\end{theorem}

\section{Directed Flow-Augmentation}\label{sec:flow-augmentation}
An \emph{instance} is a tuple $\inst = (G,s,t,k)$ where $G$ is a directed graph, $s,t \in V(G)$, and $k$ is a nonnegative integer.
An \emph{instance with a flow} is a pair $(\inst,\flow)$ where $\inst = (G,s,t,k)$ is an instance and 
$\flow = \{P_1,P_2,\ldots,P_\lambda\}$ is an $st$-flow.
In an \emph{instance with a maximum flow} we additionally require $\lambda = \lambda_G(s,t)$, that is, $\flow$ is a maximum $st$-flow.
With an instance with a flow $(\inst,\flow)$ we associate the residual graph $G^\flow$. 

In this section we prove the following results that generalize Theorems~\ref{thm:intro} and~\ref{thm:intro-det}.
\begin{theorem}\label{thm:dir-flow-augmentation}
There exists a randomized polynomial-time algorithm that, given an instance $\inst = (G,s,t,k)$, returns a set $A \subseteq V(G) \times V(G)$
and an $st$-maxflow $\witnessflow$ in $G+A$ such that 
for every star $st$-cut $Z$ of size at most $k$,
with probability $2^{-\Oh(k^4 \log k)}$ the pair $(A,\witnessflow)$ is compatible with $Z$.
\end{theorem}
\begin{theorem}\label{thm:dir-flow-augmentation-det}
There exists an algorithm that, given an instance $\inst = (G,s,t,k)$, 
      runs in time $2^{\Oh(k^4 \log k)} n^{\Oh(1)}$ and returns a set 
$\mathcal{A}$ of size $2^{\Oh(k^4 \log k)} (\log n)^{\Oh(k^3)}$
such that every element of $\mathcal{A}$ is a pair $(A,\witnessflow)$ where
$A \subseteq V(G) \times V(G)$
and $\witnessflow$ is an $st$-maxflow in $G+A$.
Furthermore, we have the following guarantee:
for every star $st$-cut $Z$ of size at most $k$
there exists a pair $(A,\witnessflow) \in \mathcal{A}$ that is compatible with $Z$.
\end{theorem}
A few remarks are in place.
First, Lemma~\ref{lem:star-cut-stays} ensures that if $A$ is compatible with $Z$, then $Z$ remains a star $st$-cut in $G+A$.
Second, we not only guarantee that $A$ is compatible with $Z$ with good probability and
that $\corecutG{Z}{G+A}$ is an $st$-mincut in $G+A$, but also that $\witnessflow$
is a witnessing flow; in particular, that a witnessing flow exists at all. 
Third, a usage of Theorem~\ref{thm:dir-flow-augmentation}  makes sense even for cuts $Z$ with $\corecut{Z}$ being an $st$-mincut
for the sake of the witnessing flow: we may need to add some edges to $G$ for a witnessing flow to  exist, let alone being returned by the algorithm. 
Finally, due to the convention that for $\lambda_{G+A}(s,t) = +\infty$, any flow containing an $st$-path with all arcs of capacity $+\infty$ is an $st$-maxflow, 
the algorithm is allowed to return a pair $(A,\witnessflow)$ such that $\lambda_{G+A}(s,t) = +\infty$ and $\witnessflow$ contains an $st$-path with all edges of capacity $+\infty$
(but clearly such a pair $(A,\witnessflow)$ is not compatible with any star $st$-cut $Z$ of size at most $k$). 

We will prove Theorems~\ref{thm:dir-flow-augmentation} and~\ref{thm:dir-flow-augmentation-det}
in parallel, as the proofs are mostly the same. The leading narration will be for the randomized case, as we find it more natural.
In most cases, the difference is that in the randomized case we perform a random guess, and in the 
deterministic case we branch. This becomes a bit more tricky when the random step is 
of color-coding type (cf. Section~\ref{ss:color-coding}), in particular when one combines
outputs of recursive calls from different areas highlighted by color-coding.

If $\lambda_G(s,t) = 0$, then for every star $st$-cut $Z$ in $G$ we have $\corecut{Z} = \emptyset$, so we can return $A = \emptyset$ and $\flow=\emptyset$ in Theorem~\ref{thm:dir-flow-augmentation}
and $\mathcal{A} = \{(\emptyset,\emptyset)\}$ in Theorem~\ref{thm:dir-flow-augmentation-det}. 
If $\lambda_G(s,t) > k$, then there is no star $st$-cut $Z$ of size at most $k$, so $A = \emptyset$ and $\flow$ being any $st$-maxflow is a valid outcome for Theorem~\ref{thm:dir-flow-augmentation}
while $\mathcal{A} = \emptyset$ is a valid outcome for Theorem~\ref{thm:dir-flow-augmentation-det}.
Henceforth we assume $0 < \lambda_G(s,t) \leq k$.

We set a threshold $\ellthreshold := 4k^2 + 3$ for later use. Note that while the algorithm is recursive and the value of $k$ may decrease in the recursive calls,
the threshold $\ellthreshold$ is set once at the beginning, using the initial value of $k$.

\subsection{Reachability patterns, leaders, and mincut sequences}
Let $(\inst = (G,s,t,k),\flow = \{P_1,\ldots,P_\lambda\})$ be an instance with a maximum flow. 

A \emph{reachability pattern} is a directed graph $H$ with vertex set $V(H) = [\lambda]$, that is, each vertex $i \in V(H)$ corresponds
to a path $P_i \in \flow$, that contains a self-loop at every vertex.
The \emph{pattern associated with $(\inst,\flow)$} is a graph $H$ with $V(H) = [\lambda]$ and $(i,j) \in E(H)$ if and only
if there exists $v \in V(P_i) \setminus \{s,t\}$ and $u \in V(P_j) \setminus \{s,t\}$ such that there is a $v$ to $u$ path in $G^\flow$. 
Note that $H$ is a reachability pattern if and only if every path $P_i$ has at least one internal vertex. 

For a vertex $v \in V(G)$, the set $\resreach{v}$ is the set of vertices reachable from $v$ in $G^\flow$. 
For a vertex $v \in V(G)$ and an index $i \in [\lambda]$, the \emph{last vertex on $P_i$ residually reachable from $v$}, denoted
$\lastreach{v}{i}$, is the last (closest to $t$) vertex $u$ on $P_i$ that is reachable from $v$ in $G^\flow$.
Note that if $u$ is reachable from $v$ in $G^\flow$, then all vertices preceeding $u$ on $P_i$ are also reachable from $v$ in $G^\flow$, that is,
the set of vertices of $P_i$ reachable from $v$ in $G^\flow$ form a prefix of $P_i$. 
Furthermore, note that 
if $v_1$ and $v_2$ are on $P_j$ and $v_1$ is earlier than $v_2$ on $P_j$, then $\lastreach{v_1}{i}$ is not later
than $\lastreach{v_2}{i}$ on $P_i$ (they may be equal).

We have the following simple observation.
\begin{lemma}\label{lem:resreach}
Let $\flow$ be an $st$-maxflow in instance $\inst=(G,s,t,k)$ and assume $\lambda_G(s,t) > 0$. 
Then for every $v \in V(G)$ exactly one of the following options hold: 
\begin{itemize}
\item $s \in \resreach{v}$, $t \notin \resreach{v}$, and $\delta_G^+(\resreach{v})$ is an $st$-mincut;
\item $s,t \in \resreach{v}$, $v$ is on the $t$-side of every $st$-mincut, and $\delta_G^+(\resreach{v}) = \emptyset$;
\item $s,t \notin \resreach{v}$, $t$ is not reachable from $v$ in $G$, and $\delta_G^+(\resreach{v}) = \emptyset$.
\end{itemize}
\end{lemma}
\begin{proof}
Note that $s\in \resreach{v}$ if and only if there is a path from $v$ to a vertex on $\flow$ in $G$. The 
forward implication follows from that any directed 
path from $v$ to $V(\flow)$ in $G^\flow$ contains as a prefix a subpath from $v$ to $V(\flow)$ which is also a path in $G$. 
The backward implication is trivial.

To begin with, consider the case when no vertex on $\flow$ is reachable from $v$ in $G$. 
From the observation in the previous paragraph, it follows that no vertex on $\flow$ is reachable from $v$ in $G^\flow$ as well. 
In particular, we have $s,t \notin \resreach{v}$. 
It remains to note that any edge contained in $\delta_G^+(\resreach{v})$ is on $\flow$ since otherwise the head of the edge would have 
been included in $\resreach{v}$. From $\resreach{v}\cap V(\flow)=\emptyset$, we conclude $\delta_G^+(\resreach{v}) = \emptyset$.

Therefore, we turn to the case when some vertex on $\flow$ is reachable from $v$ in $G$, which subsumes $s\in \resreach{v}$. 
First, if $t\notin \resreach{v}$, observe that the set of vertices on $P_i$ reachable from $v$ in $G^\flow$ forms a proper prefix of $P_i$ for every $i\in[\lambda]$. 
Therefore, $\delta_G^+(\resreach{v})$ takes precisely one edge from each $P_i\in \flow$. 
It suffices to show that $C:=\delta_G^+(\resreach{v})$ is an $st$-cut. Indeed, an $st$-path in $G-C$ contains 
a path in $G-E(\flow)$ from a vertex on $P_i$ before $C\cap E(P_i)$ to a vertex on $P_j$ after $C\cap E(P_j)$ for some $i,j\in [\lambda]$. 
That is, a vertex on $\flow$ after $C\cap E(\flow)$ is reachable from $v$ in $G^\flow$, a contradiction.
Second, if $t\in \resreach{v}$, then $\resreach{v}$ contains the entire vertices on $\flow$ and thus $\delta_G^+(\resreach{v}) = \emptyset$. 
That $v$ is on the $t$-side of every $st$-mincut immediately follows from the condition $t\in \resreach{v}$.
\end{proof}

Let $C$ be an $st$-mincut in $\inst$ and let $H$ be a reachability pattern.
For $i \in [\lambda]$, the \emph{leader of the path $P_i$ after $C$} (with respect to the pattern $H$), denoted $\leader{H}{C}{i}$, is the first (closest to $s$) vertex
$v$ on $P_i$ such that for every $(i,j) \in E(H)$ the vertex $\lastreach{v}{j}$ is after $C$ on $P_j$. 
A few remarks are in place. First, the notion of the leader is well-defined as the vertex $t$ is always a feasible candidate. 
Second, since a reachability pattern is required to contain a self-loop at every vertex, for every $i \in [\lambda]$ there is at least
one edge $(i,j) \in E(H)$ to consider.
Third, as $C$ is oriented from the $t$-side to the $s$-side in $G^\flow$, $\leader{H}{C}{i}$ is after $C$ on $P_i$ and, furthermore,
the entire path in $G^\flow$ from $\leader{H}{C}{i}$ to $\lastreach{\leader{H}{C}{i}}{j}$ for every $(i,j) \in E(H)$ 
lies on the $t$-side of $C$.

For an $st$-mincut $C$ and a reachability pattern $H$, we define the \emph{mincut $H$-subsequent to $C$} as follows.
If $t \in \resreach{\leader{H}{C}{i}}$ for some $i \in [\lambda]$, the mincut $H$-subsequent to $C$ is undefined.
Otherwise, Lemma~\ref{lem:resreach} implies that $\delta^+(\resreach{\leader{H}{C}{i}})$ is an $st$-mincut for every $i \in [\lambda]$. 
Let $X = \bigcup_{i \in [\lambda]} \resreach{\leader{H}{C}{i}}$. By submodularity, $C' := \delta^+(X)$ is also an $st$-mincut;
we proclaim $C'$ to be the mincut $H$-subsequent to $C$. 
Observe that, by definition, for every $i \in [\lambda]$, the leader $\leader{H}{C}{i}$ lies in the $t$-side of $C$ and $s$-side of $C'$,
the set $\resreach{\leader{H}{C}{i}}$ lies in the $s$-side of $C'$ and for every $(i,j) \in E(H)$ any path in $G^\flow$ from $\leader{H}{C}{i}$ 
to a vertex on $P_j$ in the $t$-side of $C$ (in particular, to $\lastreach{\leader{H}{C}{i}}{j}$) lies entirely in the $t$-side of $C$ and $s$-side of $C'$.

For a reachability pattern $H$, an \emph{$H$-sequence of mincuts} is a sequence $C_1,C_2,\ldots,C_\ell$ of mincuts defined as follows.
$C_1$ is the $st$-mincut closest to $s$ and for $a > 1$ the mincut $C_a$ is the mincut $H$-subsequent to $C_{a-1}$, as long as it is defined. 
Observe that for every $1 \leq a < b \leq \ell$ and $i \in [\lambda]$ we have that the edge of $C_a$ on $P_i$ lies strictly before the edge of $C_b$
on $P_i$. That is, the $s$-side of $C_a$ is contained in the $s$-side of $C_b$ and, furthermore, on every path $P_i$
the $s$-side of $C_a$ is a strict subset of the $s$-side of $C_b$. See Figure~\ref{fig:seq} for an illustration. 

For $a\in [\ell-1]$, let $G'$ be the graph obtained from $G$ by contracting the $s$-side of $C_a$ and the $t$-side of $C_{a+1},$ 
and $\flow'$ be the maximum $st$-flow of $G'$ obtained by shortening each $P_i$ so as to start with the edge $E(P_i)\cap C_a$ 
and to end with the edge $E(P_i)\cap C_{a+1}$ for $i\in [\lambda]$. 
Observe that the pattern associated with $(\inst',\flow')$ is precisely the patter associated with $(\inst,\flow).$ 

\begin{figure}[tb]
\begin{center}
\begin{tikzpicture}[scale=1.25]
\draw (-5,0) node {$s$};
\draw (5,0) node {$t$};
\foreach \y in {1,2,3,4} {
  \draw[-latex] (-4.5, 1.25-\y*0.5) -- (4.5, 1.25-\y*0.5) node[pos=0,above] {\tiny $P_\y$};
}
\foreach \n/\x in {1/-4.25, 2/-2.5, 3/-0.5, 4/2.5} {
  \draw[very thick, red] (\x, 1) -- (\x, -1);
  \draw (\x, 1.25) node {$C_\n$};
}

\foreach \x/\y/\z in {-4/2/4, -3.5/1/2, -3/3/4} {
  \draw[-latex] (\x,1.25-\y*0.5) -- (\x,1.25-\z*0.5);
}
\foreach \x/\y/\z in {-2/1/4, -1.5/2/4, -1/1/2, -1/3/4} {
  \draw[-latex] (\x,1.25-\y*0.5) -- (\x,1.25-\z*0.5);
}
\foreach \x/\y/\z in {0/1/2, 0.5/2/4, 1/3/4, 1.5/3/4, 2/1/4} {
  \draw[-latex] (\x,1.25-\y*0.5) -- (\x,1.25-\z*0.5);
}
\foreach \x/\y/\z in {3/1/2, 3.5/1/2, 4/1/2, 3.5/3/4, 4/3/4} {
  \draw[-latex] (\x,1.25-\y*0.5) -- (\x,1.25-\z*0.5);
}

\tikzstyle{vertex}=[circle,fill=black,minimum size=0.2cm,inner sep=0pt]

\draw (-4.5, -2) node {Pattern $H$:};
\node[vertex] (p1) at (-3, -2) {};
\node[vertex] (p2) at (-1, -2) {};
\node[vertex] (p3) at (1, -2) {};
\node[vertex] (p4) at (3, -2) {};
\foreach \n in {1,2,3,4} {
  \draw[below] (p\n) node {$P_\n$};
  \draw[-to] (p\n) edge[out=150,in=180] ($ (p\n) + (0,0.3) $);
  \draw[] ($ (p\n) + (0,0.3) $) edge[out=0,in=30] (p\n);
}
\draw [-latex] (p1) -- (p2);
\draw [-latex] (p3) -- (p4);
\draw [-latex] (p1) edge[out=20,in=160] (p4);
\draw [-latex] (p2) edge[out=10,in=170] (p4);
\end{tikzpicture}
\caption{A schematic example of an $H$-sequence of mincuts.
  The cuts are marked red. The last cut needs to be that late, because a connection
from $P_1$ to $P_4$ is missing. Recall that the reachability is checked in the residual graph, where the flow
paths are traversed backwards.}\label{fig:seq}
\end{center}
\end{figure}
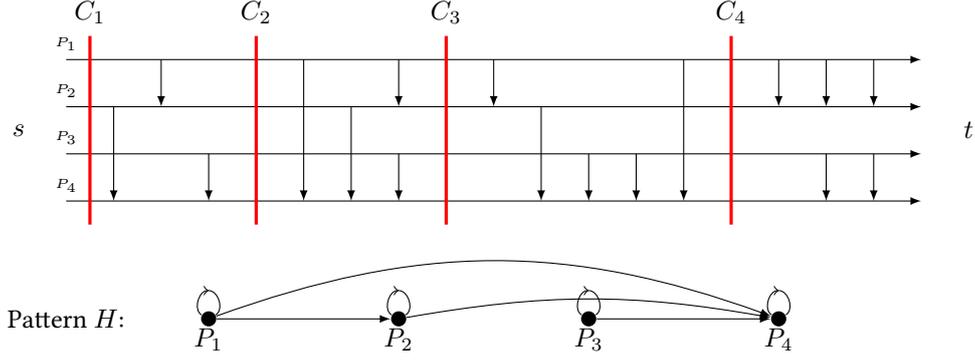

We say that the instance $\inst$ has \emph{proper boundaries} if both
$\delta^+(s)$ and $\delta^-(t)$ are $st$-mincuts and $\delta^+(s) \cap \delta^-(t) = \emptyset$, that is,
there is no arc $(s,t)$. 
In other words, $\inst$ has proper boundaries if $\delta^+(s)$ is the $st$-mincut closest to $s$, $\delta^-(t)$ is the $st$-mincut closest to $t$,
 and these cuts are disjoint.
Note that in particular if $\inst$ has proper boundaries, then the pattern $H$ associated with it is a reachability pattern.
We have the following observations.
\begin{lemma}\label{lem:proper-pattern}
Let $(\inst,\flow)$ be an instance with a maximum flow such that $\inst$ has proper boundaries and
let $H$ be the associated reachability pattern. 
Let $C_1,C_2,\ldots,C_\ell$ be the $H$-sequence of mincuts.
Then $C_1 = \delta^+(s)$ and $\ell \geq 2$, that is, $C_2$ is defined.
\end{lemma}
\begin{proof}
$C_1 = \delta^+(s)$ by the definition of having proper boundaries.
By the definition of $H$, for every $(i,j) \in E(H)$ there exists a vertex $v_{i,j} \in V(P_i) \setminus \{s,t\}$
and $u_{i,j} \in V(P_j) \setminus \{s,t\}$ such that $u_{i,j}$ is reachable from $v_{i,j}$ in $G^\flow$. 
Consequently, $\leader{H}{C_1}{i} \neq t$ for every $i \in [\lambda]$. 
Since $\delta^-(t)$ is an $st$-mincut, $t \notin \resreach{\leader{H}{C_1}{i}}$ for every $i \in [\lambda]$.
Thus, $C_2$ is defined.
\end{proof}

\begin{lemma}\label{lem:trans-pattern}
Let $(\inst,\flow)$ be an instance with a maximum flow such that $\inst$ has proper boundaries and
let $H$ be the associated reachability pattern. 
Let $C_1,C_2,\ldots,C_\ell$ be the $H$-sequence of mincuts.
If $\ell \geq 3$, then $H$ is transitive. 
\end{lemma}
\begin{proof}
Let $(i_1,i_2)$ and $(i_2,i_3)$ be two arcs of $H$. Our goal is to prove that $(i_1,i_3) \in E(H)$.

By the construction of $C_3$, since $(i_1,i_2) \in E(H)$, 
there exists $v \in V(P_{i_1})$ and $u \in V(P_{i_2})$ such that both $v$ and $u$ are after $C_2$ and before $C_3$
on paths $P_{i_1}$ and $P_{i_2}$, respectively,
and $u$ is reachable from $v$ in $G^\flow$.
Similarly, by the construction of $C_2$, since $(i_2,i_3) \in E(H)$, 
there exists $v' \in V(P_{i_2})$ and $u' \in V(P_{i_3})$ such that both $v'$ and $u'$ are after $C_1$ but before $C_2$
on paths $P_{i_2}$ and $P_{i_3}$, respectively,
and $u'$ is reachable from $v'$ in $G^\flow$.
Since $u$ is after $C_2$ on $P_{i_2}$ and $v'$ is before $C_2$ on $P_{i_2}$, $v'$ is before $u$ on $P_{i_2}$ and hence $v'$ is reachable from $u$
in $G^\flow$. 
We infer that $u' \in V(P_{i_3}) \setminus \{s,t\}$ is reachable from $v \in V(P_{i_1}) \setminus \{s,t\}$ in $G^\flow$, hence $(i_1,i_3) \in E(H)$,
as desired.
\end{proof}

\subsection{Recursion structure and initial steps}

The algorithms of Theorem~\ref{thm:dir-flow-augmentation}
and of Theorem~\ref{thm:dir-flow-augmentation-det} are recursive.
The input is an instance with a flow $(\inst = (G,s,t,k),\flow = \{P_1,\ldots,P_\lambda\})$ and an integer $\cutlb$.
In the randomized case, the goal is to return a set $A \subseteq V(G) \times V(G)$ and an $st$-maxflow $\witnessflow$ in $G+A$
such that $\lambda_{G+A}(s,t) \geq \cutlb$ 
and 
for every star $st$-cut $Z$ with $|Z| \leq k$ and $|\corecutG{Z}{G}| \geq \cutlb$, 
$(A,\witnessflow)$ is compatible with $Z$
(i.e., $A$ is compatible with $Z$, $\corecutG{Z}{G+A}$ is an $st$-mincut in $G+A$, and $\witnessflow$ is a witnessing flow for $Z$ in $G+A$)
with good probability. 
In the deterministic case, we want a small family $\mathcal{A}$ of pairs $(A,\witnessflow)$
with $\lambda_{G+A}(s,t) \geq \cutlb$ and $\witnessflow$ being an $st$-maxflow in $G+A$,
     such that for every star $st$-cut $Z$ with $|Z| \leq k$ and $|\corecutG{Z}{G}| \geq \cutlb$,
at least one element of $\mathcal{A}$ is compatible with $Z$.

For an input $(\inst,\flow,\cutlb)$, the algorithm may perform some randomized choices / branching steps and a number of recursive calls. 
In each case, we will argue that the returned output (one of the returned outputs) $(A,\witnessflow)$ is compatible with $Z$, usually using the fact that a similar property holds for outputs of recursive subcalls.
While performing random choices or branching steps, we will always say what to aim for (e.g., for a vertex $v$ we may randomly guess if it is on the $s$-side or $t$-side of the hypothethical cut $Z$)
  and bound the probability of a correct choice. 
Furthermore, for some recursive calls $(\inst',\flow',\cutlb')$ we will define a star $st$-cut $Z'$ (depending usually on $Z$)
and prove that if the returned output $(A',\witnessflow')$ is compatible with $Z'$, then the final output $(A,\witnessflow$) is compatible with $Z$ (if the returned family $\mathcal{A}'$ contains
    an element compatible with $Z'$, then the final output $\mathcal{A}$ contains an element
   compatible with $Z$ respectively in the deterministic case).

However, for sake of clarity in the description of the algorithm we will not perform a formal probability analysis.
(We perform a formal probability analysis for Theorem~\ref{thm:dir-flow-augmentation}
 in Section~\ref{ss:prob-analysis}
 and an analysis of the size of the output family $\mathcal{A}$ for Theorem~\ref{thm:dir-flow-augmentation-det} in Section~\ref{ss:det-analysis}.)
Instead, we always informally indicate what is the measure of progress in the recursive subcalls, so that the claim that the success probability in the randomized case is lower bounded by a function
of $k$ would be clear (but the exact estimate on the success probability requires some tedious calculations). 

A meticulous reader can notice that the flow $\flow$ is not used in the definitions above. 
We will use it to keep track of the progress of the algorithm; in some recursive calls the only progress
will be that the flow $\flow$ at hand changes structure in some sense. 

\medskip

Let us proceed now to the description of the algorithm.

Given $(\inst,\flow,\cutlb)$, the algorithm first performs a number of preprocessing steps.
First, if $\flow$ is not an $st$-maxflow, we compute an $st$-maxflow $\flow'$ and recurse on $(\inst,\flow',\cutlb)$. 
Thus, we henceforth assume $\flow$ is an $st$-maxflow, $\lambda = \lambda_G(s,t)$. 

If $\lambda = 0$, then every star $st$-cut $Z$ satisfies $\corecut{Z} = \emptyset$. 
In this case we return $A = \emptyset$ and $\witnessflow=\emptyset$ if $\cutlb = 0$
and $A = \{(s,t)\}$ and $\witnessflow$ consisting of one flow path along the edge $(s,t)$ if $\cutlb > 0$. 
(The deterministic algorithm returns a family $\mathcal{A}$ with a single element being the
 pair $(A,\witnessflow)$ as above.) Notice that if $\lambda=0$ and $\cutlb>0$, any star $st$-cut $Z$ 
 has $\corecut{Z}=\emptyset$, hence there is no $Z$ that the output set $A\subseteq V(G)\times V(G)$ needs to be compatible with.
  
If $\lambda > k$, then there is no star $st$-cut $Z$ with $|Z| \leq k$.
Hence, we can return $A = \{(s,t)\}$ and $\witnessflow$ consisting of one flow path along $(s,t)$. 
(Again, in the deterministic case, we return $\mathcal{A} = \{(A,\witnessflow)\}$ for
 $A$ and $\witnessflow$ as above.)

In the remaining (and most interesting) case $0 < \lambda = \lambda_G(s,t) \leq k$, we start by computing
$C_\leftarrow$ and $C_\rightarrow$, the $st$-mincuts closest to $s$ and $t$, respectively.
We can also assume $\cutlb \geq \lambda$, as we can replace $\cutlb := \max(\cutlb,\lambda)$. 

If $C_\leftarrow \neq \delta^+(s)$, we proceed as follows.
Let $A_\leftarrow$ be the set of tails of edges of $C_\leftarrow$, except for $s$. 
In the randomized case, we make a random guess: with probability $0.5$ we guess that there exists a vertex $v \in A_\leftarrow$
that is on the $t$-side of $Z$, and with the remaining probability we guess that there is no such vertex.
The guess is correct with probability $0.5$. 

In the first case, we additionally guess one tail $v \in A_\leftarrow$ that is on the $t$-side of $Z$
(there are $|A_\leftarrow| \leq \lambda \leq k$ options, so we are correct with probability at least $k^{-1}$). 
Let $G' = G + \{(v,t)\}$ and note that $\lambda_{G'}(s,t) > \lambda_G(s,t)$ as $C_\leftarrow$ is the $st$-mincut closest to $s$. 
We recurse on $((G',s,t,k),\flow,\cutlb)$, obtaining a pair $(A',\witnessflow')$ and return $(A' \cup \{(v,t)\}, \witnessflow')$. 
Clearly, if the guess is correct, $Z$ remains a star $st$-cut in $G'$ and, if furthermore the output $(A',\witnessflow')$ is compatible with $Z$ in $G'$,
  then $(A' \cup \{(v,t)\}, \witnessflow')$ is compatible with $Z$ in $G$. 
The probability that we correctly entered this case is at least $(2k)^{-1}$, and in the recursive call the value $\lambda_G(s,t)$ increased.

(In the deterministic case, we replace the guess with branching in a standard manner:
 we invoke a branch for each $v \in A_\leftarrow$
 that assumes that $v$ is in the $t$-side of $Z$ and recurse as above, and as a last branch
 proceed with the second case as in the next paragraph.)

In the second case, we define $G'$ as $G$ with the whole $s$-side of $C_\leftarrow$ contracted onto $s$, recurse on $G'$
(with the flow paths of $\flow$ shortened to start from the edges of $C_\leftarrow$), 
obtaining in the randomized case 
a pair $(A',\witnessflow')$, and return $A := A' \cup \{(s,v)~|~v \in A_\leftarrow\}$ and $\witnessflow$ being the flow $\witnessflow'$ with every path potentially prepended with the 
appropriate edge $(s,v)$, $v \in A_\leftarrow$.
(In the deterministic case, we obtain a set $\mathcal{A}'$ from the recursive call
and perform the above modification to every element of $\mathcal{A}'$.)
If all tails of edges of $C_\leftarrow$ are on the $s$-side of $Z$, then we did not contract any edge of $\corecut{Z}$,
$\corecut{Z}$ remains a minimal $st$-cut in $G'$ and $Z' := Z \cap E(G')$ remains a star $st$-cut in $G'$. 
It is straightforward to observe if
furthemore $(A',\witnessflow')$ is compatible with $Z'$ in $G'$, then $(A,\witnessflow)$ is compatible with $Z$ in $G$. 
Furthermore, in the recursive call the $st$-mincut closest to $s$ in $G'$ is $\delta^+(s)$ (being the image of $C_\leftarrow$).

We perform a symmetric process if $C_\rightarrow \neq \delta^-(t)$. 
That is, let $A_\rightarrow$ be the set of heads of edges of $C_\rightarrow$, except for $t$. 
If there is a $v \in A_\rightarrow$ that is in the $s$-side of $Z$, we guess so (with probability 0.5), guess $v$ (with probability at least $k^{-1}$),
recurse on $G' = G + \{(s,v)\}$, obtaining a pair $(A',\witnessflow')$ and return $(A' \cup \{(s,v)\}, \witnessflow')$. 
As before, in the recursive call $\lambda_{G'}(s,t) > \lambda_G(s,t)$ as $C_\rightarrow$ is the $st$-mincut closest to $t$.
Otherwise, we guess that this is the case (with probability 0.5) and contract the $t$-side of $C_\rightarrow$ onto $t$. 
We observe that if the guess is correct, then no edge of $Z$ is contracted onto $t$, $Z$ remains a star $st$-cut in the resulting graph $G'$.
Similarly as before, we recurse on $G'$, obtaining a pair $(A',\witnessflow')$ and return $A := A' \cup \{(v,t)~|~t \in A_\rightarrow\}$
and $\witnessflow$ constructed from $\witnessflow'$ by appending an appropriate arc $(v,t)$, $v \in A_\rightarrow$ to some of the paths. 
In the recursive call the $st$-mincut closest to $t$ is $\delta^-(t)$ (being the image of $C_\rightarrow$). The deterministic counterpart of the above process is fully analogous. 

As a result, we either already recursed and returned an answer
(being a correct guess with probability $\Omega(k^{-1})$ in the randomized case
 or invoking $\Oh(k)$ recursive calls in the deterministic case)
or in the input graph, $\delta^+(s)$ and $\delta^-(t)$ are $st$-mincuts. 

As a last preprocessing step, we check if $G$ contains an arc $(s,t)$. If this is the case, then clearly $(s,t) \in Z$ for any $st$-cut $Z$ 
and $(s,t)$ is one of the paths in $\flow$, say $P_i$. 
We recurse on $G' := G-\{(s,t)\}$ with $\flow \setminus \{P_i\}$, parameter $k-1$ instead of $k$, and $\cutlb-1$ instead of $\cutlb$, obtaining a pair $(A',\witnessflow')$ in the randomized case.
Furthermore, $Z \setminus \{(s,t)\}$ remains a star $(s,t)$-cut in $G'$. 
Thus, we can safely return $(A',\witnessflow' \cup \{P_i\})$ in this case.
(In the deterministic case, we obtain the set $\mathcal{A}'$ from recursion and apply the
 same modification to every $(A',\witnessflow') \in \mathcal{A}'$). 

\medskip

Hence, we are left with $\inst$ being an instance with proper boundaries. 
We compute the reachability pattern $H$ of $(\inst,\flow)$ and the $H$-sequence of mincuts $C_1,C_2,\ldots,C_\ell$.
We have $C_1 = \delta^+(s)$ and, by Lemma~\ref{lem:proper-pattern}, $\ell \geq 2$.
We split into three cases, tackled in the next three subsections:
\begin{itemize}
\item $|E(H)| = |V(H)|$, that is, $H$ contains only self-loops at every vertex;
\item $\ell \leq \ellthreshold$, that is, the $H$-sequence $C_1,\ldots,C_\ell$ is short;
\item $\ell > \ellthreshold$, that is, the $H$-sequence $C_1,\ldots,C_\ell$ is long. 
\end{itemize}

\subsection{Base case: only self-loops in $H$}
We now deal with the base case $|E(H)| = |V(H)|$, that is, $H$ consist of $\lambda$ vertices with self-loops and nothing else.

We have the following observation.
\begin{lemma}\label{lem:base-case-structure}
Let $(\inst,\flow)$ be an instance with maximum flow with proper boundaries and let $H$ be its reachability pattern.
Assume $|V(H)| = |E(H)|$. 
For every $i \in [\lambda]$, $v \in V(P_i) \setminus \{t\}$, and a bottleneck edge $e$ on $P_i$, if $e$ lies before $v$ on $P_i$, then any path from $s$ to $v$ in $G$
visits $e$.
Consequently, any set consisting of one bottleneck edge from each path $P_i$ is an $st$-mincut in $G$.
\end{lemma}
\begin{proof}
For the first claim, let $Q$ be a path from $s$ to $v$ in $G$ avoiding $e$. 
By the assumption on proper boundaries, $t \notin V(Q)$. If $Q$ contains a vertex of $V(P_j) \setminus \{s,t\}$ for some $j \neq i$, then 
a minimal subpath of $Q$ from a vertex of $\bigcup_{j \neq i} V(P_j) \setminus \{s,t\}$ to a vertex of $V(P_i) \setminus \{s,t\}$ is a path also in $G^\flow$ and thus 
contradicts the assumption that $H$ contains only self-loops.

Let $C$ be an $st$-mincut witnessing that $e$ is a bottleneck edge. Since $e$ is before $v$ on $P_i$, the concatenation of $Q$ and a subpath of $P_i$ from $v$ to $t$ is 
an $st$-path disjoint with $C$, a contradiction. This finishes the proof of the first claim.

For the second claim, pick a bottleneck edge $e_i \in E(P_i)$ on every path $P_i$ and let $Y = \{e_i~|~i\in [\lambda]\}$. 
Assume $Y$ is not an $st$-mincut.
This implies that there exists a path $Q$ in $G$ that starts on a path $P_i$ in a vertex $v$ before $e_i$, ends on a path $P_j$ in a vertex $u$
after $e_j$ and does not contain any edge of $Y$ nor any internal vertex on paths of $\flow$.
In particular, $Q$ does not contain any edge of $\flow$.

Note that if $v \neq s$ and $u \neq t$, then $(i,j) \in E(H)$ by the definition of the graph $H$.
Hence, either $i=j$, or $v=s$, or $u=t$. Note that if $v=s$ or $t=u$ we could have chosen $i=j$ anyway, so 
it suffices to consider the case $i=j$. 
Prepend $Q$ with a subpath of $P_i$ from $s$ to $v$, obtaining a path $Q'$ from $s$ to $u$ avoiding $e_i$. This contradicts the first claim.
\end{proof}
\begin{lemma}\label{lem:base-case}
Let $(\inst,\flow)$ be an instance with maximum flow with proper boundaries and let $H$ be its reachability pattern.
Assume $|V(H)| = |E(H)|$ and let $Z$ be a star $st$-cut in $G$.
Then exactly one of the following cases hold.
\begin{itemize}
\item There exists $i \in [\lambda]$ such that no edge of $Z \cap E(P_i)$ is a bottleneck edge and $\corecut{Z}$ is not an $st$-mincut, that is, $|\corecut{Z}| > \lambda_G(s,t)$.
\item $\corecut{Z}$ is an $st$-mincut, that is, $|\corecut{Z}| = \lambda_G(s,t)$, and no edge of $Z \setminus \corecut{Z}$ is a bottleneck edge of $G$.
\end{itemize}
\end{lemma}
\begin{proof}
Assume first there is $i \in [\lambda]$ such that no edge of $Z \cap E(P_i)$ is a bottleneck edge. Then, $Z$ contains no $st$-mincut, in particular, $\corecut{Z}$ is not an $st$-mincut, and the first option holds.

In the other case, for every $i \in [\lambda]$, pick the first (closest to $s$) bottleneck edge $e_i \in Z \cap E(P_i)$ and denote $Z' = \{e_i~|~i \in [\lambda]\}$.
By Lemma~\ref{lem:base-case-structure}, $Z'$ is an $st$-mincut. 
Furthermore, for every $i \in [\lambda]$, every edge $e$ on $P_i$ after $e_i$ is in the $t$-side of $Z'$. Since $Z$ is a star $st$-cut, such an edge $e$ is not in $Z$; in particular, $Z'$ are the only bottleneck edges in $Z$.
We infer that for every $i \in [\lambda]$ we have $e_i \in \corecut{Z}$, that is, $Z' \subseteq \corecut{Z}$. As $Z'$ is an $st$-cut, $Z' = \corecut{Z}$ and the second case holds.
\end{proof}

We proceed as follows with the randomized algorithm. If $\cutlb = \lambda$, with probability 0.5 guess that $\corecut{Z}$ is an $st$-mincut in $G$. 
In this case, for every $i \in [\lambda]$ proceed as follows. Let $(u_{i,1},v_{i,1}), \ldots, (u_{i,a_i},v_{i,a_i})$ be the bottleneck edges of $P_i$, in the order of their appearance on $P_i$.
Denote $v_{i,0} = s$, $u_{i,a_i+1} = t$, $A_i = \{(v_{i,b},u_{i,b+1})~|~0 \leq b \leq a_i\}$ and $P_i'$ to be the path consisting of edges $(s,u_{i,1})$, $(u_{i,1},v_{i,1})$, $(v_{i,1},u_{i,2})$, \ldots, $(u_{i,a_i},v_{i,a_i})$, $(v_{i,a_i},t)$. We return $(A:=\bigcup_{i\in[\lambda]} A_i, \witnessflow = \{P_i'~|~i \in [\lambda]\})$. 
By Lemma~\ref{lem:base-case}, $(A,\witnessflow)$ is compatible with $Z$.

If $\cutlb > \lambda$, or $\cutlb = \lambda$ but (with the remaining probability 0.5) we guessed that $\corecut{Z}$ is not an $st$-mincut in $G$,
we randomly guess $i \in [\lambda]$, aiming at $Z \cap E(P_i)$ containing no bottleneck edge. 
We set $A_0$ to be the copies of all bottleneck edges of $P_i$, recurse on $G' := G+A_0$, obtaining a pair $(A',\witnessflow')$, and return $(A:=A' \cup A_0,\witnessflow')$. 
Clearly, if the guess is correct, $A_0$ is compatible with $Z$ and $Z$ remains an star $st$-cut in $G'$. Furthermore, if $(A',\witnessflow')$ is compatible with $Z$ in $G'$, then $(A,\witnessflow')$ is compatible
with $Z$ in $G$.
Since we added a copy of every bottleneck edge on one flow path, $\lambda_{G'}(s,t) > \lambda_G(s,t)$. 
This finishes the description of this case.

The deterministic counterpart of the above process is the natural one.
First, we insert to the constructed family $\mathcal{A}$ the pair $(A,\witnessflow)$ from the first case above. Then, we branch into $\lambda$ cases, one for each $i \in [\lambda]$, 
recurse as in the second case above, obtaining a set $\mathcal{A}'$, and insert into 
$\mathcal{A}$ a pair $(A' \cup A_0,\witnessflow')$ for every $(A',\witnessflow') \in \mathcal{A}'$.

\subsection{Small $\ell$ case}
In the next case we assume $|E(H)| > |V(H)|$ and $\ell \leq \ellthreshold$. Recall $\ellthreshold := 4k^2 + 3$.

For clarity, we describe only the randomized case of Theorem~\ref{thm:dir-flow-augmentation} 
below, as all guesses can be replaced by branching steps in a straightforward manner
(in particular, we do not use any color-coding steps here).

If there is no star $st$-cut $Z$ of size at most $k$ with $|\corecut{Z}| \geq \cutlb$,
then any pair $(A,\witnessflow)$ with $\lambda_{G+A}(s,t) \geq \cutlb$ is a valid outcome, so we can assume at least one such star $st$-cut exists.
For the sake of analysis, fix one star $st$-cut $Z$ with $|Z| \leq k$ and $|\corecut{Z}| \geq \cutlb$. 

Let $B = \bigcup_{i=1}^\ell V(C_i)$, that is, $B$ is the set of endpoints of the cuts $C_1,\ldots,C_\ell$. 
Let $B_\leftarrow$ be the set of elements of $B$ on the $s$-side of $Z$ and $B_\rightarrow$ be the set of vertices of $B$ on the $t$-side of $Z$.

For every $i \in [\lambda]$, we guess the last vertex $v_i$ of $B_\leftarrow$ (there is always one, as $s$ is a candidate) 
  and the first vertex $u_i$ of $B_\rightarrow$ on $P_i$ (or guess $u_i = \bot$ meaning that there is no such vertex).  
There are at most $2\ellthreshold(2\ellthreshold+1)$ options per index $i$, so in total the success probability is $2^{-\Oh(k \log k)}$.

We now consider a number of corner cases. We use the first applicable corner case, if possible,
so that in subsequent corner cases we can assume that the earlier ones are not applicable.

\paragraph{Corner case 1: $B_\rightarrow = \emptyset$, that is, for every $i \in [\lambda]$ we guessed $u_i = \bot$.}
In other words, all endpoints of edges of $\bigcup_{a=1}^\ell C_a$ are on the $s$-side of $Z$. 
In particular, all endpoints of $C_\ell$ are on the $s$-side of $Z$ and $t$ is not a head of any edge of $C_\ell$.
Let $G'$ be the graph $G$ with the $s$-side of $C_\ell$ contracted onto $s$.
Then, $\corecut{Z}$ is a minimal $st$-cut in $G'$ as well, $Z' := Z \cap E(G')$ is a star $st$-cut in $G'$,
and the paths $P_i$, shortened to start from the edge of $C_\ell$, form a maximum $st$-flow $\flow_{G'}$ in $G'$. 
Furthermore, as $G$ has proper boundaries, so has $G'$ (the cut $C_\ell$ becomes $\delta_{G'}^+(s)$ in $G'$). 
Let $H'$ be the reachability pattern of $G'$ and $\flow_{G'}$. 
We claim the following.
\begin{claim}
$E(H') \subsetneq E(H)$. 
\end{claim}
\begin{proof}
Clearly $E(H') \subseteq E(H)$, as any path in the residual graph of $G'$ and $\flow_{G'}$ witnessing $(i,j) \in E(H')$ is also a path in $G^\flow$ and thus witnesses $(i,j) \in E(H)$.

We claim that if $E(H') = E(H)$, then $C_{\ell+1}$ would have been defined. Indeed, from the definition of $H'$ for every $(i,j) \in E(H')$ there exist
vertices $v_{i,j} \in V(P_i) \setminus \{t\}$ and $u_{i,j} \in V(P_j) \setminus \{t\}$ that are after $C_\ell$ on their $P_i$ and $P_j$ respectively
such that $u_{i,j}$ is reachable from $v_{i,j}$ in the residual graph of $G'$ and $\flow_{G'}$ (and hence also in $G^\flow$). 
If $E(H') = E(H)$, this witnesses that $\leader{H}{C_\ell}{i} \neq t$ for every $i \in [\lambda]$. 
Since $\delta_G^-(t)$ is an $st$-mincut in $G$ 
and thus $t$ is not reachable from any vertex on $V(P_i)\setminus \{t\}$ in $G^\flow$ for every $i\in [\lambda]$,
we have $t \notin \resreach{\leader{H}{C_\ell}{i}}$ for every $i \in [\lambda]$.  
Therefore, $C_{\ell+1}$ would have been defined, a contradiction. 
\end{proof}
We recurse on $G'$, $\flow_{G'}$, $k$, and $\cutlb$ obtaining a pair $(A',\witnessflow')$.
We set 
$A_0 = \{(s,v)~|~v \in B\}$
and return $(A := A_0 \cup A', \witnessflow)$, where  every flow path of $\witnessflow$ is constructed from a flow path of $\witnessflow'$ by potentially prepending it with an appropriate edge of $A_0$. 
Clearly, if $(A',\witnessflow')$ is compatible with $Z'$ in $G'$, then $(A,\witnessflow)$ is compatible with $Z$ in $G$.

\paragraph{Corner case 2: for some $i \in [\lambda]$, there is a vertex of $B_\rightarrow$ before a vertex of $B_\leftarrow$ on $P_i$.}
That is, for some $i \in [\lambda]$ the guessed first vertex $u_i \in B_\rightarrow \cap V(P_i)$ is before the guessed last vertex $v_i \in B_\leftarrow \cap V(P_i)$.
Define $A_0 = \{(s,v_i),(u_i,t)\}$ and observe that $A_0$ is compatible with $Z$ and in $G+A_0$ there is an augmenting path for $\flow$: 
start with the edge $(s,v_i) \in A_0$, go along reversed $P_i$ to $u_i$, and finish with the edge $(u_i, t) \in A_0$. 
Hence, we can just recurse on $G' := G+A_0$, $k$, and $\cutlb$, obtaining a pair $(A',\witnessflow')$ and return $(A' \cup A_0, \witnessflow')$. 

\medskip

Note that if Corner case 2 does not happen, from the guesses we can deduce whole sets $B_\leftarrow$ and $B_\rightarrow$: for every $i \in [\lambda]$, all vertices of $B \cap V(P_i)$
up to the (guessed) last vertex of $B_\leftarrow$ are in $B_\leftarrow$ and the remainder of $B \cap V(P_i)$ is in $B_\rightarrow$. 
Hence, we can define 
\[A_0 = \{(s,v)~|~v \in B_\leftarrow\} \cup \{(u,t)~|~u \in B_\rightarrow\}.\]
Note that $A_0$ is compatible with $Z$ (if the guesses are correct). 

In $G+A_0$, we define an $st$-flow $\flow' = \{P_i'~|~[\lambda]\}$ as follows: for every $i \in [\lambda]$, the flow path $P_i'$ consists of the arc $(s,v_i) \in A_0$, a subpath of $P_i$ from $v_i$ to $u_i$, 
and the arc $(u_i,t) \in A_0$. 

\paragraph{Corner case 3: $\lambda_{G+A_0}(s,t) > \lambda_{G}(s,t)$.}
In this case we can just recurse on $G' := G+A_0$, $\flow'$, $k$, and $\cutlb$, obtaining a pair $(A',\witnessflow')$, and return $(A' \cup A_0, \witnessflow')$. 

If we are not in this corner case, we have that $\flow'$ is an $st$-maxflow in $G+A_0$. 

\paragraph{Corner case 4: there is $i \in [\lambda]$ such that $P_i'$ contains exactly one edge $e$ of capacity $1$.}
Then clearly $e \in \corecut{Z}$. 
We can recurse on the graph $G' := G+A_0-\{e\}$, the flow $\flow' \setminus \{P_i'\}$, the parameter $k-1$ instead of $k$, and the parameter $\cutlb-1$ instead of $\cutlb$, obtaining a pair $(A',\witnessflow')$,
and return $(A := A' \cup A_0, \witnessflow' \cup \{P_i'\})$. 

Two remarks are in place. First, this corner case covers the case where $Z \cap \bigcup_{a=1}^\ell C_a \neq \emptyset$. 
Second, if we are not in this corner case, then in particular the whole $C_1$ is in the $s$-side of $Z$ as $C_1 = \delta^+(s)$. 

\paragraph{Main case}
We are left with the main case, where none of the aforementioned corner cases occur.

Since on every $P_i$ the arcs of $C_1,C_2,\ldots,C_\ell$ appear in this order, by the excluded Corner case 1 at least one head of an arc of $C_\ell$ is in $B_\rightarrow$. 
By the excluded Corner case 4, the tail of the said arc of $C_\ell$ is also in $B_\rightarrow$.

Let $a \in [\ell]$ be maximum such that all endpoints of $C_a$ are in $B_\leftarrow$. This is well-defined as $C_1 = \delta_G^+(s)$ and
all endpoints of $C_1$ are in $B_\leftarrow$ as we are not in the Corner case 4. 
Also, by the assumption of this main case, $a < \ell$, that is, $C_{a+1}$ is defined.

Let $C$ be the $st$-mincut in $G+A_0$ closest to $t$. Since $A_0$ features arcs from $s$ to all endpoints of $C_a$, $C$ lies entirely in the $t$-side of $C_a$. 
By the definition of $a$, there is an arc of $C_{a+1}$ whose endpoints lie in $B_\rightarrow$ and thus have arcs to $t$ in $A_0$. Hence, $C$ contains
at least one arc whose endpoints are on the $t$-side of $C_a$ and on the $s$-side of $C_{a+1}$. 
Furthermore, clearly $C$ is an $st$-mincut in $G$ as well (but of course not necessarily closest to $t$) because Corner case 3 is excluded.  

We randomly guess whether there is a vertex $v \in V(C)$ that is in the $s$-side of $Z$. 
With probability $0.5$ we guess that there is such a vertex.
Then, with probability $0.5$ we guess if there is a head of an edge of $C$ that is in the
$s$-side of $Z$.
If this is the case, we guess such a head $v \in V(C)$ (at most $\lambda$ options) and otherwise
we guess a tail $v \in V(C)$ that is in the $s$-side of $Z$ (again, at most $\lambda$ options). 
In the end, we obtain either an edge $(v,u) \in C$ with $v$ on the $s$-side of $Z$ and $u$ on the $t$-side of $Z$ (thus $(v,u) \in Z$)
  or a head $v$ of an edge of $C$ that is in the $s$-side of $Z$. The probability of a correct guess is at least $(4\lambda)^{-1}$. 

In the first case, we recurse on $G' := G-\{(v,u)\}$ with flow $\flow\setminus \{P_i\}$ where $P_i$ is the flow path containing $(v,u)$ and the parameter $k-1$, obtaining a pair $(A',\witnessflow')$.
It is straightforward to check that then we can return $(A := A' \cup \{(s,v),(u,t)\}, \witnessflow := \witnessflow' \cup \{((s,v), (v,u), (u,t)) \})$.

In the second case, we recurse on $G' := G + (A_0 \cup \{(s,v)\})$, $k$, and $\cutlb$, obtaining a pair $(A',\witnessflow')$. 
Note that $A_0 \cup \{(s,v)\}$ is compatible with $Z$, if the guess is correct.
Hence, we can return $(A_0 \cup \{(s,v)\} \cup A', \witnessflow')$. 
Furthermore, as $C$ is the closest to $t$ mincut of $G+A_0$, we have $\lambda_{G'}(s,t) > \lambda$.

Thus we are left with the most interesting case,  guessed with probability $0.5$, where all vertices of $V(C)$ are in the $t$-side of $Z$. 
The crucial observation now is the following.
\begin{claim}\label{cl:late-leader}
There exists $i \in [\lambda]$ such that $\leader{H}{C_a}{i}$ lies in the $t$-side of $C$. 
\end{claim}
\begin{proof}
By the definition of $a$, there exists $e \in C_{a+1}$ with an endpoint in $B_\rightarrow$. 
By the excluded Corner case 4, both endpoints of $e$ are in $B_\rightarrow$; let $v$ be the tail of $e$. Clearly $v$ is on the $t$-side of $C$ because $(v,t)\in A_0$.

By the definition of $C_{a+1}$, there exists $i \in [\lambda]$ such that $v \in \resreach{\leader{H}{C_a}{i}}$. 
Again by the definition of $C_{a+1}$, the whole $\resreach{\leader{H}{C_a}{i}}$ is on the $s$-side of $C_{a+1}$, in particular
$\leader{H}{C_a}{i}$ is. 
Since $C$ is an $st$-mincut in $G$ as well, while $v$ is reachable from $\leader{H}{C_a}{i}$ in $G^\flow$ and $v$ is on the $t$-side of $C$,
$\leader{H}{C_a}{i}$ is also on the $t$-side of $C$. 
This finishes the proof.
\end{proof}

Let $G'$ be constructed from $G$ by contracting first the $s$-side of $C_a$ onto $s$ and then the $t$-side of $C$ onto $t$. 
Note that both $C_a$ and $C$ are $st$-mincuts in $G$ and on every path $P_i$ the arc of $C_a$ appears strictly before the arc of $C$. 
Hence, $G'$ has proper boundaries and the paths of $P_i$ naturally shorten to a maximum $st$-flow $\flow_{G'}$ in $G'$. 
We recurse on $G'$, $\flow_{G'}$, and $k$, obtaining a pair $(A',\witnessflow')$.
We return $A = A_0 \cup A' \cup \{(v,t)~|~v \in V(C)\}$ and a flow $\witnessflow$ constructed from flow paths $\witnessflow'$ by potentially prepending them
with an edge $(s,v) \in A_0$ for $v$ being a head of an edge of $C_a$ or adding at the end an edge $(v,t) \in A$ for $v$ being a tail of an edge of $C$. 

Consider a star $st$-cut $Z$ in $G$. 
Assume we guessed correctly that all vertices of $B_\leftarrow$ are in the $s$-side of $Z$ and vertices of $B_\rightarrow$ and all endpoints of $C$ are in the $t$-side of $Z$.
Then, $\corecut{Z}$ is a minimal $st$-cut in $G'$ and $Z' := Z \cap E(G')$ is a star $st$-cut in $G'$.
Furthemore, if $(A',\witnessflow')$ is compatible with $Z'$ in $G'$, then $(A,\witnessflow)$ is compatible with $Z$ in $G$, as desired.

Let $H'$ be the reachability pattern of $G'$ and $\flow_{G'}$. To bound the success probability, the following observation is crucial.
\begin{claim}
$E(H') \subsetneq E(H)$.
\end{claim}
\begin{proof}
Clearly, any path witnessing $(i,j) \in E(H')$ in the residual network of $G'$ and $\flow_{G'}$ is also a path in $G^\flow$ so it witnesses $(i,j) \in E(H')$. 
We need to exhibit an element of $E(H) \setminus E(H')$. 

Let $i \in [\lambda]$ be the index asserted by Claim~\ref{cl:late-leader}. 
Since $v := \leader{H}{C_a}{i}$ lies in the $t$-side of $C$, while the head of the edge of $C_a \cap E(P_i)$ lies in the $s$-side of $C$, 
the predecessor $v'$ of $v$ on $P_i$ lies in the $t$-side of $C_a$. There is a reason why $v'$ is before $\leader{H}{C_a}{i}$ on $P_i$:
there exists $(i,j) \in E(H)$ such that $\lastreach{v'}{j}$ is before $C_a$ on $P_j$.
Since $v$ is in the $t$-side of $C$, we have that for every $u \in V(P_i) \cap V(G') \setminus \{s,t\}$ it holds that $\lastreach{u}{j}$ is before $C_a$ on $P_j$.
Consequently, $(i,j) \notin E(H')$, as desired. 
\end{proof}
Hence, in the recursive call $|E(H')|$ decreases, which is the progress giving the desired bound on the success probability.
(For formal proof of the success probability lower bound, see Section~\ref{ss:prob-analysis};
 for the corresponding analysis of the number of subcases if the guesswork is replaced
 with branching, see Section~\ref{ss:det-analysis}.)

\subsection{Large $\ell$ case}
In the remaining case we have $|E(H)| > |V(H)|$ and $\ell > \ellthreshold$.
By Lemma~\ref{lem:trans-pattern}, $H$ is transitive. 
Also, $|E(H)| > |V(H)|$ implies that $\lambda \geq 2$.

In this section the narrative is led by the randomized case of Theorem~\ref{thm:dir-flow-augmentation}, but due to color coding steps, derandomization is nontrivial in a few places and 
mandates some discussion.

In the argumentation, we will frequently use the predicate of being in the $s$-side or $t$-side of a cut $C_a$.
For brevity, we extend this notion to all indices $a \in \mathbb{Z}$, not only $a \in [\ell]$. 
That is, every vertex and edge of $G$ is in the $t$-side of $C_a$ for $a \leq 0$ and in the $s$-side of $C_a$ for $a > \ell$
while neither a vertex nor an edge of $G$ is in the $s$-side of $C_a$ for $a \leq 0$ and in the $t$-side of $C_a$ for $a > \ell$.

Fix a star $st$-cut $Z$ of size at most $k$ such that $|\corecut{Z}| \geq \cutlb$.
Let $Z_{ts}$ be the set of those edges $(v,u) \in \bigcup_{i=1}^\lambda E(P_i)$ such that
$v$ is on the $t$-side of $Z$ while $u$ is on the $s$-side of $Z$. 
Note that for every $i \in [\lambda]$ we have $|Z_{ts} \cap E(P_i)| \leq |Z \cap E(P_i)|-1$ and hence $|Z_{ts}| \leq k-\lambda$.

We say that an index $a \in \mathbb{Z}$ is \emph{touched} if there exists $e \in Z \cup Z_{ts}$ with an endpoint 
that is in the $t$-side of $C_a$ and in the $s$-side of $C_{a+1}$.
Note that a touched index $a$ satisfies $0 \leq a \leq \ell$.

We have the following observation.
\begin{claim}\label{cl:bound-touched}
There are at most $2|Z \cup Z_{ts}| \leq 4k-2\lambda$ touched indices.
\end{claim}
\begin{proof}
Every endpoint $v$ of an edge $e \in Z \cup Z_{ts}$ gives raise to at most one index $a$ being touched, namely the maximum $a \in \mathbb{Z}$ such that $v$ is in the $t$-side of $C_a$. 
Hence, there are at most $2|Z \cup Z_{ts}| \leq 4k-2\lambda$ touched indices as desired.
\end{proof}

A set $L \subseteq V(H)$ is \emph{downward-closed} (in $H$) if there is no arc $(i,j) \in E(H)$ with $i \in L$ but $j \notin L$. 
For a set $L \subseteq V(H)$, let $\closure{L}$ be the minimal superset of $L$ that is downward-closed; since $H$ is transitive and contains a loop at every vertex, we have
\[\closure{L} = \{i \in [\lambda]~|~\exists_{j \in L} (j,i) \in E(H)\}.\]
For $a \in [\ell]$, let $L_a \subseteq [\lambda]$ be the set of those indices for which both endpoints
of the unique edge of $C_a \cap E(P_i)$ lie in the $s$-side of $Z$. We also define $L_a = [\lambda]$ for $a \leq 0$ and $L_a = \emptyset$ for $a > \ell$.

We have establish a few properties of the objects defined above.
\begin{claim}\label{cl:no-escape}
Let $L \subseteq [\lambda]$ be downward closed in $H$.
Then $G$ admits no path 
\[\mathrm{from\ } \bigcup_{i \in L} V(P_i) \setminus \{s,t\}\mathrm{\ to\ }
\bigcup_{i \in [\lambda] \setminus L} V(P_i) \setminus \{s,t\}.\]
\end{claim}
\begin{proof}
Assume the contrary, let $Q$ be a shortest such path. By minimality, no internal vertex of $Q$ lies in
$\bigcup_{i \in [\lambda]} V(P_i) \setminus \{s,t\}$. 
Since $G$ has proper boundaries, neither $s$ nor $t$ lies on $Q$. Hence, $Q$ is a path in $G^\flow$, contradicting the definition of $H$ and the assumption that $L$ is downward-closed.
\end{proof}

\begin{claim}\label{cl:closure-locked}
Let $a$ be an integer. Then, for every $i \in [\lambda] \setminus \closure{L_a}$, every vertex on $P_i$ in the $t$-side of $C_a$ is also in the $t$-side of $Z$.

In particular, if $b \geq a$ is an integer, then $L_b \subseteq \closure{L_a}$, and if additionally $L_a$ is downward-closed, then $L_b \subseteq L_a$.
\end{claim}
\begin{proof}
The claim is obvious if $a < 1$ (as then $L_a = [\lambda]$) and if $a > \ell$ (as then no vertex is in the $t$-side of $C_a$), so assume otherwise.

For every $i \in [\lambda] \setminus \closure{L_a}$, pick an endpoint $v_i$ of the unique edge of $C_a \cap E(P_i)$ that is in the $t$-side of $Z$
and let $X_i$ be the set of vertices of $V(P_i) \setminus\{t\}$ that lie after $v_i$ on $P_i$. 
By Claim~\ref{cl:no-escape} applied to $L := \closure{L_a}$, for every vertex $v \in\bigcup_{i \in [\lambda] \setminus \closure{L_a}} X_i$,
any path from $s$ to $v$ of $G$ passes though $\{v_i~|~i \in [\lambda] \setminus \closure{L_a}\}$. As $v_i$ is in the $t$-side of $Z$, any such path contains an edge of $Z$.
We infer that such $v$ is in the $t$-side of $Z$.
This proves the main part of the claim.

The second part of the claim is trivial if $b > \ell$ (as then $L_b = \emptyset$), and follows directly from the first part if $1 \leq a \leq b \leq \ell$.
\end{proof}

\begin{claim}\label{cl:closure-reached}
Assume $a<b$ are two integers such that $b-a \geq \lambda-1$ and for every $a \leq c < b$, $c$ is untouched. 
Then, $L_{b} = \closure{L_a}$. 
In particular, $L_{b}$ is downward-closed.
\end{claim}
\begin{proof}
Clearly, $L_b \subseteq \closure{L_a}$ by Claim~\ref{cl:closure-locked}.
Also, since every $c$ with $a \leq c < b$ is untouched, for every $i \in L_a$ no edge of $P_i$ between the edge of $C_a$ and the edge of $C_b$ is in $Z$, 
so $L_a \subseteq L_b$. 

If $a < 1$, then $L_a = [\lambda]$ and we have $L_b = [\lambda]$ as well. 
If $b > \ell$, then $L_b = \emptyset$ so from $L_a \subseteq L_b$ is follows that $L_a = \emptyset$ as well.

Hence, we can assume $1 \leq a < b \leq \ell$.
Since $L_a \subseteq L_b \subseteq \closure{L_a}$, it suffices to show that if for some $1 \leq c < \ell$ we
have that $c$ is untouched and $L_c$ is not downward-closed, then $L_{c+1}$ is a strict superset of $L_c$.

Clearly, as $c$ is untouched $L_c \subseteq L_{c+1}$. 
Since $L_c$ is not downward-closed, there is $(i,j) \in E(H)$ with $i \in L_c$ but $j \notin L_c$. 
By the construction of $C_{c+1}$, there exists $v \in V(P_i)$ and $u \in V(P_j)$, both in the $t$-side of $C_c$ and $s$-side of $C_{c+1}$
and a path $Q$ in $G^\flow$ from $v$ to $u$.
Pick $i$, $j$, and $Q$ as above so that the length of $Q$ is minimum possible. 

Since in $G^\flow$ one cannot cross neither $C_c$ nor $C_{c+1}$ from the $s$-side to the $t$-side,
all vertices of $Q$ are in the $t$-side of $C_c$ and $s$-side of $C_{c+1}$.
By the minimality of $Q$, no internal vertex of $Q$ lies on any path $P_\iota$, $\iota \in [\lambda]$. 
Hence, $Q$ is a path in $G$ as well. 
Since $c$ is untouched, no edge of $Q$ lies in $Z$. 
Also, no edge of $P_i$ nor $P_j$ between $C_c$ and $C_{c+1}$ (including $C_{c+1}$) lies in $Z$.
We infer that $j \in L_{c+1}$, as desired. This finishes the proof of the claim.
\end{proof}

\begin{claim}\label{cl:paths-before-downward-closed}
Assume $a \in [\ell]$ is such that $L_a$ is downward-closed. Let $i \in [\lambda]$ and $e$ be the unique edge of $C_a \cap E(P_i)$. 

\begin{itemize}
\item
If $i \notin L_a$, then the following holds. The entire suffix of $P_i$ from the head of $e$ to $t$ is in the $t$-side of $Z$.
In particular, no edge of the said suffix is in $Z$.
\item 
If $i \in L_a$, then the following holds. The entire prefix of $P_i$ from $s$ to the tail of $e$ is in the $s$-side of $\corecut{Z}$.
In particular, no edge of the said prefix is in $\corecut{Z}$. 
\end{itemize}
\end{claim}
\begin{proof}
The first claim follows directly from Claim~\ref{cl:closure-locked}, as we assume that $L_a$ is downward-closed.

For the second claim, we proceed by contradiction. Assume there is a vertex on $P_i$ that is in the $t$-side of $\corecut{Z}$ and $s$-side of $C_a$.
Let $v$ be the earliest (on $P_i$) such vertex.
Since $v$ is in the $t$-side of $\corecut{Z}$, $v \neq s$. 
By the choice of $v$, the predecessor of $v$ on $P_i$ is in the $s$-side of $Z$, that is, 
the edge on $P_i$ with the head $v$ is in $\corecut{Z}$.
By the minimality of $\corecut{Z}$, there is a path $Q$ from $v$ to $t$ in $G-\corecut{Z}$. 

Since $v$ is in the $s$-side of $C_a$, there is an edge of $C_a$ on $Q$. Assume $j \in [\lambda]$ is such that 
the edge of $E(P_j) \cap C_a$ lies on $Q$. 
Since $Q$ is a path in $G-\corecut{Z}$ ending in $t$, it is entirely in the $t$-side of $\corecut{Z}$. 
By the definition of $L_a$, we have $j \notin L_a$. 
However, as $v \in V(P_i)$ and $i \in L_a$, the subpath of $Q$ from $v$ to the edge of $E(P_j) \cap C_a$ contradicts Claim~\ref{cl:no-escape}.
\end{proof}

\begin{claim}\label{cl:support-path-milestone}
Let $a \in [\ell]$ be such that for every $a - (\lambda-1) \leq c < a$ the index $c$ is untouched. 
Then, for every $i \in L_a$ there exists in $G-Z$ a path from $s$ to the tail of the unique edge of $C_a \cap E(P_i)$ whose all vertices lie in the $s$-side of $C_a$.
\end{claim}
\begin{proof}
If $a - (\lambda-1) \leq 1$ (i.e., $a \leq \lambda$), then there is no endpoint of the edge of $Z$ in the $s$-side of $C_a$ (note that if $0$ is touched, then so is $1$) 
  and hence the prefix of $P_i$ from $s$ to $C_a$
satisfies the desired properties.
We henceforth assume then that $b := a-(\lambda-1) > 1$. 

Let $e = (u,v)$ be the unique edge of $C_a \cap E(P_i)$. Since $i \in L_a$, both $u$ and $v$ are in the $s$-side of $Z$. Let $Q$ be a path from $s$ to $u$ in $G-Z$.

Since $Q$ ends in $u$, $Q$ contains an edge of $C_b$. Let $e'$ be the first edge of $C_b$ on $Q$ and let 
$j \in [\lambda]$ be such that $e' \in E(P_j) \cap C_b$. Since $Q$ is in the $s$-side of $Z$, we have $j \in L_b$.

For $b < c \leq a$, let $L_c'$ be the family of indices $i' \in [\lambda]$ such that there is a path from the head of $e'$
to the tail of the unique edge of $C_c \cap E(P_{i'})$ with all vertices in the $t$-side of $C_b$ and $s$-side of $C_c$. 
We also define $L_b' = \{j\}$. 
Since for every $b \leq c < a$, $c$ is untouched, we have that any path witnessing $i' \in L_c'$ is contained completely in the $s$-side of $Z$
and hence $L_c' \subseteq L_c$. Also, $L_c' \subseteq L_{c'}'$ for any $b \leq c \leq c' \leq a$. 

Note that it suffices to prove that $i \in L_c'$ for some $b \leq c \leq a$, as then the desired path can be formed by concatenating the prefix of $Q$ up to $e'$, 
the path witnessing $i \in L_c'$, and the part of $P_i$ from $C_{c}$ to $u$. 
As $b = a-(\lambda-1)$, it suffices to show that if $i \notin L_c'$ for some $b \leq c <a$, then $L_{c+1}'$ is a proper superset of $L_c'$. 

If $j=i$ then we are done, so assume that $j\neq i$. Fix $b \leq c < a$ such that $i \notin L_c'$. 
As $j \in L_c'$ and $e' \in E(P_j)$, there is a subpath of $Q$, between $e'$ and $u$, that starts in a vertex of some $P_{j'}$, $j' \in L_c'$, ends in a vertex of some $P_{i'}$, $i' \notin L_c'$,
and contains no internal vertices on paths of $\flow$. This path witnesses that $(j',i') \in E(H)$ and thus there is a path $R$ in $G^\flow$ from a vertex of $P_{j'}$ to a vertex of $P_{i'}$
that lies completely in the $t$-side of $C_c$ and $s$-side of $C_{c+1}$. 
A path $R$ contains a subpath $R'$ with no internal vertices on $\flow$ that starts on a path $P_{j''}$ for some $j'' \in L_c'$ and ends on a path $P_{i''}$ for some $i'' \notin L_{c}'$. 
This path is present in $G$. Consequently, by using a subpath of $P_{j''}$ from $C_c$ to the start of $R'$, the path $R'$, and a subpath of $P_{i''}$ from
the end of $R'$ to $C_{c+1}$ we witness that $i'' \in L_{c+1}'$, as desired. 
\end{proof}

Claims~\ref{cl:closure-reached},~\ref{cl:paths-before-downward-closed}, and~\ref{cl:support-path-milestone}
motivate the following definition. An index $a \in [\ell]$ is a \emph{milestone} if for every $a-(\lambda-1) \leq c \leq a$, $c$ is untouched. 
Claim~\ref{cl:closure-reached} implies for a milestone $a$, the set $L_a$ is downward-closed.

An arc $(v,u)$ is a \emph{long backarc} if there exists a milestone $a$ with $v$ in the $t$-side of $C_a$ and $u$ in the $s$-side of $C_a$. 
Not all assertions of the following observations are used later, but we present them anyway as they give good intuition about why the
long backarcs are essentially irrelevant for the problem.
\begin{claim}\label{cl:long-backarcs}
Let $\longbackarcs$ be the family of long backarcs. Then,
\begin{itemize}
\item $\flow$ is (still) a maximum $st$-flow in $G-\longbackarcs$;
\item $G-\longbackarcs$ is (still) properly boundaried;
\item the reachability pattern of $(G-\longbackarcs,\flow)$ is (still) $H$;
\item $C_1,C_2,\ldots,C_\ell$ is (still) an $H$-sequence of mincuts in $(G-\longbackarcs,H)$;
\item $\corecut{Z}$ is disjoint with $\longbackarcs$ and is (still) a minimal $st$-cut in $G-\longbackarcs$;
\item $Z' \subseteq Z$, consisting of those arcs of $Z \setminus \longbackarcs$ whose tail is reachable from $s$ in $G-(Z \cup \longbackarcs)$, contains $\corecut{Z}$ and 
  is a star $st$-cut in $G-\longbackarcs$.
\end{itemize}
\end{claim}
\begin{proof}
The first two properties are immediate as $\longbackarcs$ does not contain any edge of $\flow$. 
For the third and fourth properties, recall that $\ell > \ellthreshold > 2$ and notice that when defining $C_a$ for $1 < a \leq \ell$, all 
paths in $G^\flow$ that witness where $\leader{H}{C_{a-1}}{i}$ is for $i \in [\lambda]$ are contained in the $t$-side of $C_{a-1}$ and in the $s$-side of $C_a$.
Hence, these paths do not use any edge of $\longbackarcs$.

For the penultimate property, clearly $Z \setminus \longbackarcs$ is an $st$-cut in $G-\longbackarcs$.
Let $e \in \corecut{Z}$. Since $\corecut{Z}$ is a minimal $st$-cut, there exists a path $Q$ from $s$ to $t$ whose only intersection with $\corecut{Z}$ is $e$. 
Pick such a path $Q$ that minimizes the number of edges of $\longbackarcs$.
To prove the minimality of $\corecut{Z}$ in $G-\longbackarcs$, it suffices to show that there is no long backarc on $Q$ (which in particular implies that $e \notin \longbackarcs$). 

To this end, we use Claim~\ref{cl:paths-before-downward-closed}.
Assume that $Q$ contains at least one edge of $\longbackarcs$ and let $f = (v,u)$ be the last such edge. 
Let $a$ be the milestone index witnessing that $f \in \longbackarcs$. 
Since $v$ is in the $t$-side of $C_a$ and $u$ is in the $s$-side of $C_a$, $Q$ contains an edge $e_1$ of $C_a$ before $f$
and an edge $e_2$ of $C_a$ after $f$.
Let $i_1,i_2 \in [\lambda]$ be such that $e_1 \in E(P_{i_1})$ and $e_2 \in E(P_{i_2})$.

If $e_1 \notin L_a$, then $e$ is not later on $Q$ than $e_1$ and, by Claim~\ref{cl:paths-before-downward-closed}, we can substitute the suffix of $Q$ from the head of $e_1$ to $t$ with the suffix of $P_{i_1}$ from the head of 
$e_1$ to $t$, contradicting the minimality of $Q$. 
If $e_2 \in L_a$, then $e_2$ is before $e$ on $Q$ and, by Claim~\ref{cl:paths-before-downward-closed}, we can substitute the prefix of $Q$ from $s$ to the tail of $e_2$ with the prefix of $P_{i_2}$ from $s$ to the tail of $e_2$ all of whose vertices lie in the $s$-side of $C_a$,
again contradicting the minimality of $Q$. Hence, $e_1 \in L_a$ but $e_2 \notin L_a$.
However, as $e_1$ is before $e_2$ on $Q$, the subpath of $Q$ from $e_1$ to $e_2$ contradicts Claim~\ref{cl:no-escape}.

For the final property, note that the penultimate property implies $\corecut{Z} \subseteq Z'$ and thus $Z'$ is an $st$-cut. 
Since every arc of $Z \setminus Z'$ has its tail not reachable from $s$ in $G-(Z \cup \longbackarcs)$, the sets of vertices reachable from $s$
in $G-(Z \cup \longbackarcs)$ and $G-(Z' \cup \longbackarcs)$ are equal. 
Hence, for every $e \in Z'$, the tail of $e$ is reachable from $s$ and the head of $e$ is not reachable from $s$ in $G-(Z' \cup \longbackarcs)$. 
This finishes the proof that $Z'$ is a star $st$-cut.
\end{proof}
We remark here that it is possible that $\longbackarcs \cap Z \neq \emptyset$ or $Z' \neq Z \setminus \longbackarcs$; Claim~\ref{cl:long-backarcs} only asserts that the properties
  of $\corecut{Z}$ remain unchanged upon deletion of $\longbackarcs$.

We say that two indices $a$ and $b$ are \emph{close} if $|b-a| \leq \lambda$. 
Let $0 \leq a_1 < a_2 < \ldots < a_r \leq \ell$ be the touched indices; by Claim~\ref{cl:bound-touched} we have $r \leq 4k-2\lambda$. 
A \emph{block} is a maximal subsequence $a_\alpha, a_{\alpha+1}, \ldots, a_\beta$ such that $a_i$ and $a_{i+1}$ are close for every $\alpha \leq i < \beta$. 
Note that we have $a_\beta - a_\alpha \leq (\beta-\alpha)\lambda \leq (r-1)\lambda$. 

Naturally, the sequence $a_1,\ldots,a_r$ partitions into a number of blocks. 
An index $a \in \mathbb{Z}$ is \emph{interesting} if $a_\alpha - \lambda \leq a \leq a_\beta + \lambda$ for some block $a_\alpha,\ldots,a_\beta$.
Note that this condition is equivalent to $|a-a_\alpha| \leq \lambda$ for some $\alpha \in [r]$. 
We infer that the number of interesting indices is bounded by
$r \cdot (2\lambda+1) \leq 12k\lambda$.

We randomly sample a subset $\Gamma \subseteq \{0,1,\ldots,\ell\}$ in the following skewed way: every $a \in \{0,1,\ldots,\ell\}$ belongs to $\Gamma$ with probability $k^{-1}$,
   independently of the other indices. We aim at the following: for every interesting index $a$, $a \in \Gamma$ if and only
if $a$ is touched. Note that no integer outside the set $\{0,\ldots,\ell\}$ is touched, there are at most $4k-2\lambda$ touched indices and at most $12k\lambda$ interesting indices.
Thus we are successful with probability at least 
\[ k^{-(4k-2\lambda)} \cdot \left(1-k^{-1}\right)^{12k\lambda} = 2^{-\Oh(k \log k)}.\]
For the deterministic version of the above step, we use Theorem~\ref{thm:rand-separation} to obtain
a family of $2^{\Oh(k \log k)} \log n$ candidates for $\Gamma$ and branch on the choice
of $\Gamma$.

A \emph{$\Gamma$-milestone} is an index $a \in [\ell]$ such that for every $a-(\lambda-1) \leq c \leq a$, $c \notin \Gamma$. Note that if
the guess is correct, every $\Gamma$-milestone is a milestone.

A \emph{$\Gamma$-block} is a maximal sequence $b_1 < b_2 < \ldots < b_\zeta$ of elements of $\Gamma$ such that $b_i$ is close to $b_{i+1}$ 
for $1 \leq i < \zeta$. 
We perform the following cleaning procedure: for every $\Gamma$-block $b_1 < b_2 < \ldots < b_\zeta$ such that $\zeta > 4k-2\lambda$ or $b_\zeta-b_1 > (4k-2\lambda-1)\lambda$,
we delete all elements of the said block from $\Gamma$. 
If we are successful, then even after the cleaning procedure every block is a $\Gamma$-block, but there may be numerous $\Gamma$-blocks that are not blocks. 
Furthermore, if we are successful and $(b_1 < \ldots < b_\zeta)$ is a $\Gamma$-block, then both $b_1-1$ (if $b_1 > 1$) and $b_\zeta+\lambda$ (if $b_\zeta+\lambda \leq \ell$)
are $\Gamma$-milestones (and thus in particular actual milestones).

Let $\block^1, \block^2, \ldots, \block^\xi$ be the sampled $\Gamma$-blocks.
For a $\Gamma$-block $\block^\alpha = (b_1 < b_2 < \ldots < b_\zeta)$,
    we define $\leftblock^\alpha = b_1$, $\rightblock^\alpha = b_\zeta$, and 
let $Z_0^\alpha$ be the set of edges of $Z$ with both endpoints on the $t$-side of $C_{b_1}$ and $s$-side of $C_{b_\zeta+\lambda}$. 

We henceforth assume that the guess of $\Gamma$ is a correct one. 
We have the following observations.
\begin{claim}\label{cl:Z-local}
For every $e \in Z \setminus \longbackarcs$ there exists a $\Gamma$-block $\block^\alpha$ that is an actual block such 
that both endpoints of $e$ are in the $t$-side of $C_{\leftblock^\alpha}$ and in the $s$-side of $C_{\rightblock^\alpha+1}$.
\end{claim}
\begin{proof}
Let $a$ be maximum such that the head of $e$ is in the $t$-side of $C_a$. Such $a$ exists and $a \geq 1$ as $a=1$ is one of the candidates.
Clearly, $a$ is touched. Let $\block^\alpha$ be the $\Gamma$-block (and an actual block) 
where $a$ lies.
It suffices to prove that the tail of $e$ lies in the $t$-side of $C_{\leftblock^\alpha}$ and in the $s$-side of $C_{\rightblock^\alpha+1}$. 

By assumption, $e$ is not a long backarc.
If $\leftblock^\alpha > 1$ then clearly also $e$ cannot have a tail in the $s$-side of $C_{\leftblock^\alpha-1}$ as the head of $e$ is in the $t$-side of $C_{\leftblock^\alpha}$.
If $\rightblock^\alpha + \lambda \leq \ell$, then $\rightblock^\alpha+\lambda$ is a milestone and hence $e$ does not have its tail on the $t$-side of $C_{\rightblock^\alpha + \lambda}$.
We infer that actually both endpoints of $e$ are on the $t$-side of $C_{\leftblock^\alpha-1}$ and on the $s$-side of $C_{\rightblock^\alpha+\lambda}$.
For every $\rightblock^\alpha < a \leq \rightblock^\alpha + \lambda$, we have $a \notin \Gamma$, and hence $a$ is untouched as we assume the guess of $\Gamma$ is correct.
Hence, by the definition of being touched, 
neither of the endpoint of $e$ can lie at the same time in the $s$-side of $C_{\rightblock^\alpha+1}$ and in the $t$-side of $C_{\rightblock^\alpha + \lambda}$.
We infer that both endpoints of $e$ are in the $t$-side of $C_{\leftblock^\alpha}$ and in the $s$-side of $C_{\rightblock^\alpha+1}$, as desired.
\end{proof}

Claim~\ref{cl:Z-local} asserts that $(Z_0^\alpha)_{\alpha \in [\xi]}$ is a partition of $Z-\longbackarcs$. 

\begin{claim}\label{cl:L-decrease}
For every $\block^\alpha$, the following holds.
 \begin{itemize}
\item 
$L_{\leftblock^\alpha-1} = L_{\rightblock^\alpha+\lambda}$ if $\block^\alpha$ is not an actual block or $\corecut{Z} \cap Z_0^\alpha = \emptyset$
\item 
$L_{\leftblock^\alpha-1} \supsetneq L_{\rightblock^\alpha+\lambda}$ if $\block^\alpha$ is an actual block and $\corecut{Z} \cap Z_0^\alpha \neq \emptyset$.
\end{itemize}
\end{claim}
\begin{proof}
Claim~\ref{cl:closure-reached} implies that both $L' := L_{\leftblock^\alpha-1}$ and $L'' := L_{\rightblock^\alpha+\lambda}$ are downward-closed
(recall that $L_a = [\lambda]$ for $a \leq 0$ and $L_a = \emptyset$ for $a > \ell$).
Then, Claim~\ref{cl:closure-locked} implies $L'' \subseteq L'$. 

Assume first $L' \neq L''$ and pick $i \in L' \setminus L''$. Consider the path $P_i$. By Claim~\ref{cl:paths-before-downward-closed},
there is no edge of $\corecut{Z}$ on $P_i$ before $C_{\leftblock^\alpha-1}$ and no edge of $Z$ on $P_i$ after $C_{\rightblock^\alpha+\lambda}$.
As $\corecut{Z}$ is a minimal $st$-cut, 
and the last one on $P_i$ among the edges of $E(P_i)\cap Z$   is in $\corecut{Z}$, 
$P_i$ goes through $\corecut{Z}$ between $C_{\leftblock^\alpha-1}$ and $C_{\rightblock^\alpha+\lambda}$.
We infer that there is an edge of $\corecut{Z}$ between $C_{\leftblock^\alpha-1}$ and $C_{\rightblock^\alpha+\lambda}$.
By Claim~\ref{cl:Z-local}, both endpoints of the said edge is between $C_{\leftblock^\alpha}$ and $C_{\rightblock^\alpha+1}$. 
Hence, $Z_0^\alpha$ contains an edge of $\corecut{Z}$.

Assume now $L' = L''$ and, by contradiction, let $e$ be any edge of $\corecut{Z} \cap Z_0^\alpha$.
That is, both endpoints of $e$ are in the $t$-side of $C_{\leftblock^\alpha}$ and in the $s$-side of $C_{\rightblock^\alpha+1}$. 

By the definition of $\corecut{Z}$, there exists an $st$-path $Q$ whose only intersection with $Z$ is the edge $e$. 
Let $f = (v,u)$ be the first edge of $C_{\leftblock^\alpha}$ on $Q$, let $i \in [\lambda]$ be such that $f \in E(P_i)$,
and let $Q_1$ be the prefix of $Q$ from $s$ to $v$. 
Since $Q_1$ is contained in the $s$-side of $C_{\leftblock^\alpha}$, while $e$ is in the $t$-side of $C_{\leftblock^\alpha}$, $Q_1$ is a path in $G-\corecut{Z}$ as well.
Since $\leftblock^\alpha-1$ is untouched, $f \notin \corecut{Z}$ and thus $u$ is in the $s$-side of $\corecut{Z}$.
This witnesses that $i \in L_{\leftblock^\alpha}$ and hence $L' \neq \emptyset$. Since $L' = L''$, it holds that
$L_{\rightblock^\alpha+\lambda} \neq \emptyset$. In particular, $\rightblock^\alpha + \lambda \leq \ell$ and thus $C_{\rightblock^\alpha+\lambda}$ is defined.

Recall that $Q$ contains $e$. Since $e$ has its head in the $s$-side of $C_{\rightblock^\alpha+1}$, $Q$ visits at least one edge of $C_{\rightblock^\alpha+\lambda}$ after $e$ 
Let $f'$ be the last edge of $C_{\rightblock^\alpha+\lambda}$ on $Q$ and let $Q_2$ be the suffix of $Q$ starting from the head of $f'$. 
Clearly, $f'$ is after $e$ on $Q$ so $e$ does not lie on $Q_2$. 
Also, $Q_2$ is contained completely on the $t$-side of $C_{\rightblock^\alpha+\lambda}$ by the choice of $f'$. 
Hence, $Q_2$ is also a path in $G-\corecut{Z}$. Consequently, as $\corecut{Z}$ is an $st$-cut, $f'$ lies on a path $P_j$ for some $j \notin L_{\rightblock^\alpha+\lambda} = L''$. 
However, then the subpath $Q_3$ of $Q$ from the head of $f$ to the tail of $f'$ contradicts Claim~\ref{cl:no-escape}. 
This is the desired contradiction.
\end{proof}

Let $\alpha_1 < \alpha_2 < \ldots < \alpha_\eta$ be the indices of $\Gamma$-blocks that are actual blocks. 
We randomly guess the integer $1 \leq \eta \leq 4k-2\lambda$.
We sample a sequence $\mathcal{L} = ([\lambda] = L_Z^0 \supseteq L_Z^1 \supseteq \ldots L_Z^\eta = \emptyset)$ of downward-closed sets,
   aiming at $L_Z^{\iota-1} = L_{\leftblock^{\alpha_\iota}-1}$ and $L_Z^\iota = L_{\rightblock^{\alpha_\iota}+\lambda}$
   for every $1 \leq \iota \leq \eta$.
We define $J \subseteq [\eta]$ to be the set of those indices $\iota \in [\eta]$ such that $L_Z^{\iota-1} \subsetneq L_Z^\iota$.
We remark that in the deterministic setting the above step is replaced by branching 
in the straightforward manner. 

While we were able to guess sets $(L_Z^\iota)_{0 \leq \iota \leq \eta}$, we cannot guess the indices $(\alpha_\iota)_{1 \leq \iota \leq \eta}$, as there are too many options.
Instead, we perform a color-coding step, for every $\Gamma$-block guessing its potential place in the sequence of actual blocks. 
More precisely, for every $1 \leq \alpha \leq \xi-1$, we sample a set $L^\alpha \in \mathcal{L}$, and further denote
$L^0 = [\lambda]$ and $L^\xi = \emptyset$. 
We aim that for every $1 \leq \iota \leq \eta$ actually
$L^{\alpha_\iota-1} = L_Z^{\iota-1}$ and $L^{\alpha_\iota} = L_Z^\iota$. 
In the deterministic setting, we use Theorem~\ref{thm:color-coding} to replace the above
color-coding step with branching into at most 
$2^{\Oh(\eta \log \eta)} \cdot \Oh(\log n) = 2^{\Oh(k \log k)} \cdot \Oh(\log n)$ options. 

In other words, for every $1 \leq \alpha \leq \xi-1$, the sampled set $L^\alpha$
is a guess for the downward-closed set $L_{\rightblock^{\alpha}+\lambda}$
(and thus $L^{\alpha-1}$ is also a guess for the downward-closed set $L_{\leftblock^{\alpha}-1}$).
We aim at being correct in this guess around (just before and just after) $\Gamma$-blocks that
are actual blocks. 

If this is the case, then consider a single $\Gamma$-block $\block^\alpha$. 
If there is an index $1 \leq \iota \leq \eta$ such that $L^{\alpha-1} = L_Z^{\iota-1}$
and $L^\alpha = L_Z^\iota$, then it is possible that $\block^\alpha$ is an actual block
and $\alpha = \alpha_\iota$.
There may be multiple such blocks $\alpha$ for a single value of $\iota$ and we will treat
them in the same way, considering them candidates for $\block^{\alpha_\iota}$. 
On the other hand, if for an index $\alpha$ no such index $\iota$ exists,
we know that $\block^\alpha$ is not an actual block and we can add some
edges to the returned set $A$ that simply bypass $\block^\alpha$.

To implement the above intuition, we need a few notions. 
We say that an index $1 \leq \alpha \leq \xi$ is \emph{good} if there exists $1 \leq \iota \leq \eta$ such that $L^{\alpha-1} = L_Z^{\iota-1}$,
and $L^\alpha = L_Z^\iota$. 
A good index $\alpha$ is \emph{excellent} if additionally this $\iota$ belongs to $J$, that is, $L^{\alpha-1} \neq L^\alpha$. 
Note that all indices $\alpha_\iota$ for $1 \leq \iota \leq \eta$ are good 
and $\alpha_\iota$ is excellent if and only if $L_Z^{\iota-1} \neq L_Z^\iota$,
which is equivalent to $\corecut{Z} \cap Z_0^{\alpha_\iota} \neq \emptyset$ by Claim~\ref{cl:L-decrease}. 
Furthermore, if $\alpha$ is excellent, then the corresponding index $\iota \in J$ is uniquely defined and we denote it $\iota(\alpha)$. 

For every excellent $1 \leq \alpha \leq \xi$,
we define an instance $\inst^\alpha = (G^\alpha, s, t, k^\alpha)$ with a maximum flow $\flow^\alpha$ as follows. 
Let $D^\alpha = L^{\alpha-1} \setminus L^\alpha = L^{\iota(\alpha)-1}_Z \setminus L^{\iota(\alpha)}_Z$ (recall that for excellent $\alpha$ we have $L^{\alpha-1} \supsetneq L^\alpha$). 
Define first a graph $G_0^\alpha$ that consists of:
\begin{itemize}
\item a subgraph of $G$ induced by all vertices $v \in V(G)$ which are on the $t$-side 
of $C_{\leftblock^\alpha-1}$ and on the $s$-side of $C_{\rightblock^\alpha+\lambda}$, except for the vertices $s$ and $t$;
\item vertices $s$ and $t$ (but so far no arcs incident with them)
\item if $\leftblock^\alpha > 1$, then for every $i \in D^\alpha$ an edge $(s,v_i)$ where $v_i$ is the head of the edge of $E(P_i) \cap C_{\leftblock^\alpha-1}$;
\item if $\leftblock^\alpha \leq 1$, then all edges of the form $(s,v_i)$ that lie on paths $\{P_i~|~i \in D^\alpha\}$; 
\item if $\rightblock^\alpha+\lambda \leq \ell$, then for every $i \in D^\alpha$ an edge $(u_i,t)$ where $u_i$ is the tail of the edge of $E(P_i) \cap C_{\rightblock^\alpha+\lambda}$;
\item if $\rightblock^\alpha+\lambda > \ell$, then all edges of the form $(u_i,t)$ that lie on paths $\{P_i~|~i \in D^\alpha\}$
\end{itemize}
The flow $\flow^\alpha$ is of size $\lambda^\alpha := |D^\alpha|$ and consists of, for every $i \in D^\alpha$, the arc $(s,v_i)$, the subpath of $P_i$ between $v_i$ and $u_i$, and the arc $(u_i,t)$.
The graph $G^\alpha$ is the subgraph of $G_0^\alpha$ induced by all vertices that are at the same time reachable from $s$ and from which one can reach $t$ in $G_0^\alpha$.
Note that the flow $\flow^\alpha$ is present in $G^\alpha$. 

We define $Z^\alpha \subseteq Z^\alpha_0$ as the set of those $(u,v) \in Z$ such that either:
\begin{itemize}
\item $u,v \in V(G^\alpha) \setminus \{s,t\}$,
\item $\leftblock^\alpha \leq 1$ and $u=s$, $v=v_i$ for some $i \in D^\alpha$, or
\item $\rightblock^\alpha+\lambda > \ell$ and $u=u_i$, $v=t$ for some $i \in D^\alpha$.
\end{itemize}
We remark that $Z_0^\alpha \setminus Z^\alpha$ may be nonempty. For example, 
$Z_0^\alpha$ may contain an edge $e$ on a path $P_i$ for some $i \notin D^\alpha$
  that has both endpoints beween $C_{\leftblock^\alpha}$ and $C_{\rightblock^\alpha}$
and is an edge of $Z \setminus \corecut{Z}$. This edge is not in $G^\alpha$ 
and hence also not in $Z^\alpha$.

For every $\iota \in J$ we guess integers $0 \leq \cutlb_Z^\iota \leq k_Z^\iota \leq k$, aiming at $\cutlb_Z^\iota = |Z^{\alpha_\iota} \cap \corecut{Z}|$ and $k_Z^\iota = |Z^{\alpha_\iota}|$.
We mandate that for every $\iota \in J$ it holds that $|L_Z^{\iota-1} \setminus L_Z^\iota| \leq \cutlb_Z^\iota \leq k_Z^\iota$
and that $\sum_{\iota \in J} k_Z^\iota \leq k$. If these conditions are not satisfied, we just return either $(A=\emptyset, \flow)$ if $\cutlb = \lambda$ or $(A=\{(s,t)\}, \{(s,t)\})$ if $\cutlb > \lambda$
from this recursive call and terminate.
(In the deterministic setting, the above guess is replaced by branching in the standard manner,
 and we ignore branches where the conditions are not satisfied.)
Furthermore, for every excellent $1 \leq \alpha \leq \xi$, we denote $\cutlb^\alpha = \cutlb_Z^{\iota(\alpha)}$ and
$k^\alpha = k_Z^{\iota(\alpha)}$ for brevity. 
This defines the parameter $k^\alpha$ in the instance $\inst^\alpha$ for an excellent $\alpha$ and a parameter $\cutlb^\alpha$ we will pass down in the recursion.

Let us lower bound the probability that all gueses are as we aim for. 
The probability of guessing $\eta$ correctly is $\Omega(k^{-1})$. 
Guessing the sequence $\mathcal{L}$ boils down to guessing, for every $i \in [\lambda]$, the maximum index $0 \leq \iota \leq \eta$ for which $i \in L_Z^\iota$. 
Hence, the success probability here is at least $(\eta+1)^{-\lambda} = 2^{-\Oh(k \log k)}$. 
The integers $\{\cutlb_Z^\iota~|~\iota \in J\}$ and $\{k_Z^\iota~|~\iota \in J\}$ are guessed correctly with probability $2^{-\Oh(k \log k)}$
and similarly correct sets $L^\alpha$ for $\alpha \in \bigcup_{\iota=1}^\eta \{\alpha_\iota-1,\alpha_\iota\}$
are guessed with probability $2^{-\Oh(k \log k)}$.
Hence, overall success probability is $2^{-\Oh(k \log k)}$. 
By a similar computation, in the deterministic counterpart we have $2^{\Oh(k \log k)} \Oh(\log^2 n)$ branches (there were two color-coding steps so far).

We observe the following.
\begin{claim}\label{cl:progress}
For every excellent $1 \leq \alpha \leq \xi$, $\flow^\alpha$ is a maximum $st$-flow in $G^\alpha$ and $G^\alpha$ has proper boundaries.
Furthermore, one of the following options hold:
\begin{itemize}
\item either $D^\alpha \neq [\lambda]$ and thus $|\flow^\alpha| = |D^\alpha| < \lambda$ and $k^\alpha \leq k-(\lambda-|D^\alpha|) < k$, or
\item $D^\alpha = [\lambda]$, $H$ is the reachability pattern of $(G^\alpha,\flow^\alpha)$, and $\delta_{G^\alpha}^+(s), C_{\leftblock^\alpha}, C_{\leftblock^\alpha+1}, \ldots, C_{\rightblock^\alpha+\lambda-1}, \delta_{G^\alpha}^-(t)$
is the $H$-sequence of mincuts in $(G^\alpha,\flow^\alpha)$ (without the first term if $\leftblock^\alpha=1$, without the last term if $\rightblock^\alpha+\lambda > \ell$, and without all terms $C_c$ for which $c > \ell$).
\end{itemize}
\end{claim}
\begin{proof}
The claim that $\flow^\alpha$ is a maximum $st$-flow in $G^\alpha$ and that 
$G^\alpha$ is properly boundaried follows directly from how we chose edges incident with $s$ and $t$ in $G^\alpha$. 

Consider first the simpler case $D^\alpha \neq [\lambda]$. The fact that $|\flow^\alpha| = |D^\alpha| < \lambda$ is immediate.
Since we mandate $k_Z^\iota \geq |L_Z^{\iota-1} \setminus L_Z^\iota|$ and $\sum_{\iota=1}^\eta k_Z^\iota \leq k$, 
we have $k^\alpha \leq k- (\lambda-|D^\alpha|)$. 

Consider then the case $D^\alpha = [\lambda]$.
Observe that $G_0^\alpha$ is isomorphic to the graph $G$ with the following modifications: if $\leftblock^\alpha > 1$, then contract the $s$-side of $C_{\leftblock^\alpha-1}$ onto $s$
and if $\rightblock^\alpha + \lambda \leq \ell$, contract the $t$-side of $C_{\rightblock^\alpha+\lambda}$ onto $t$. Furthermore, $\flow^\alpha$ can be then defined as the projection of $\flow$ shortened
to the part in the $t$-side of $C_{\leftblock^\alpha-1}$ and in the $s$-side of $C_{\rightblock^\alpha+\lambda}$.

To see that the reachability pattern of $(G^\alpha,\flow^\alpha)$ is still $H$, first note that every path in the residual graph of $G^\alpha$ and $\flow^\alpha$ that does not use $s$ nor $t$
is also a path in $G^\flow$, and hence the residual graph of $G^\alpha$ and $\flow^\alpha$ is a subgraph of $H$.
In the second direction, since the entire part of $G$ in the $t$-side of $C_{\leftblock^\alpha}$ and $s$-side
of $C_{\leftblock^\alpha+1}$ lies in $G_0^\alpha$ (recall $\lambda \geq 2$), and $C_{\leftblock^\alpha+1}$ is the $H$-subsequent mincut to $C_{\leftblock^\alpha}$, for every $(i,j) \in E(H)$ there is a path from $V(P_i) \setminus \{s,t\}$ to $V(P_j) \setminus \{s,t\}$ in $G^\flow$
that lies entirely in the $t$-side of $C_{\leftblock^\alpha}$ and in the $s$-side of $C_{\leftblock^\alpha+1}$.
By the definition of $G^\alpha$, this path is also present in $G^\alpha$. Hence, it is also present in the residual graph of $G^\alpha$ and $\flow^\alpha$.

To see that $\delta_{G^\alpha}^+(s), C_{\leftblock^\alpha}, C_{\leftblock^\alpha+1}, \ldots, C_{\rightblock^\alpha+\lambda-1}, \delta_{G^\alpha}^-(t)$ is the $H$-sequence of mincuts in $(G^\alpha,\flow^\alpha)$, recall that
$\delta_{G^\alpha}^+(s)$ is the projection of $C_{\leftblock^\alpha-1}$ if $\leftblock^\alpha>1$ and $\delta_{G^\alpha}^-(t)$ is the projection of $C_{\rightblock^\alpha+\lambda}$ if $\rightblock^\alpha+\lambda \leq \ell$.
\end{proof}

\begin{claim}\label{cl:excellent-reachability}
Let $1 \leq \alpha \leq \xi$ be such that $\block^\alpha$ is an actual block and $\alpha$ is excellent. 
Then, for every $v \in V(G^\alpha)$, it holds that:
\begin{itemize}
\item there is a path from $s$ to $v$ in $G^\alpha-Z^\alpha$ if and only if $v$ is in the $s$-side of $Z$;
\item there is a path from $v$ to $t$ in $G^\alpha-Z^\alpha$ if and only if there is a path from $v$ to $t$ in $G-Z$. 
\end{itemize}
\end{claim}
\begin{proof}
We first prove the $\Rightarrow$ implication for the first claim.
It is immediate for $v=s$, so assume $v \neq s$. Let $Q$ be a path from $s$ to $v$ in $G^\alpha-Z^\alpha$ and let $(s,v_i)$, $i \in D^\alpha$, be the first edge of this path.
By Claim~\ref{cl:support-path-milestone}, there exists a path $Q_i$ from $s$ to $v_i$ in $G-Z$ that visits only vertices in the $s$-side of $C_{\leftblock^\alpha}$ (except for the endpoint). 
The concatenation of $Q_i$ and $Q$ witnesses that $v$ is in the $s$-side of $Z$. 

The proof for the $\Rightarrow$ implication of the second claim is very similar. 
It is immediate for $v=t$, so assume $v \neq t$. Let $Q$ be a path from $v$ to $t$ in $G^\alpha-Z^\alpha$ and let $(u_i,t)$, $i \in D^\alpha$, be the last edge of this path.
By Claim~\ref{cl:paths-before-downward-closed}, the suffix of the path $P_i$ from $C_{\rightblock^\alpha+\lambda}$ to $t$ is disjoint with $Z$ and in the $t$-side of $Z$.
Hence, the concatenation of $Q$ and this suffix is a path from $v$ to $t$ in $G-Z$. 

We now move to the $\Leftarrow$ direction for the first claim.
Let $Q$ be a path from $s$ to $v$ in $G-Z$, witnessing that $v$ is in the $s$-side of $Z$. Note that $v \neq t$. 
By the definition of $G^\alpha$, there is a path $R$ from $v$ to $t$ in $G^\alpha$; let $(u_j,t)$ be the last edge of this path.

Since $v$ is in the $t$-side of $C_{\leftblock^\alpha-1}$, $Q$ intersects $C_{\leftblock^\alpha-1}$; let $e_1$ be the last edge of $C_{\leftblock^\alpha-1}$ on $Q$
and let $i_1 \in [\lambda]$ be such that $e_1$ lies on $P_{i_1}$. Since both endpoints of $e_1$ are in the $s$-side of $Z$, $i_1 \in L^{\alpha-1}$. 

We observe that no vertex $w$ on $Q$ can lie on a path $P_i$ for $i \in L^\alpha$.
If $w$ lies on both $Q$ and $P_i$ for $i \in L^\alpha$, then the subpath of $Q$ from $w$ to $v$ and the subpath of $R$ from $v$ to $u_j$ contradicts Claim~\ref{cl:no-escape} as $j \in D^\alpha$ so $j \notin L^\alpha$ and $L^\alpha$ is downward-closed. 
This in particular implies that $i_1 \in D^\alpha$. 

In particular, $Q$ contains no edge of $C_{\rightblock^\alpha+\lambda}$, as all edges of $C_{\rightblock^\alpha+\lambda}$ that are in the $s$-side of $Z$ lie on paths $P_i$, $i \in L^\alpha$. 
Consequently, the subpath $Q'$ of $Q$ from the head of $e_1$ to $v$ is completely contained in the $t$-side of $C_{\leftblock^\alpha-1}$ and the $s$-side of $C_{\rightblock^\alpha+\lambda}$. 
Concatenating $Q'$ with $R$ witnesses that the whole $Q'$ is in fact present in $G^\alpha$. 
Hence, $Q'$, prepended with $(s,v_{i_1})$ is a path from $s$ to $v$ in $G^\alpha-Z^\alpha$, as desired. 

We are left with the $\Leftarrow$ direction for the second claim, which is similar to the previous argumentation. 
Let $Q$ be a path from $v$ to $t$ in $G-Z$; note that $v \neq s$.
By the definition of $G^\alpha$, there is a path $R$ from $s$ to $v$ in $G^\alpha$; let $(s,v_j)$ be the first edge of this path.

Since $v$ is in the $s$-side of $C_{\rightblock^\alpha+\lambda}$, $Q$ intersects $C_{\rightblock^\alpha+\lambda}$; let $e_2$ be the first edge of $C_{\rightblock^\lambda-1}$ on $Q$ and let $i_2 \in [\lambda]$ be such that $e_2$ lies in $P_{i_2}$. We have $i_2 \notin L^\alpha$. 

We observe that no vertex $w$ on $Q$ can lie on a path $P_i$ for $i \notin L^{\alpha-1}$. 
If $w$ lies on both $Q$ and $P_i$ for $i \notin L^{\alpha-1}$, then the subpath of $R$ from $v_j$ to $v$ and the subpath of $Q$ from $w$ to $v$ contradicts Claim~\ref{cl:no-escape} as $j \in D^\alpha$ so $j \in L^{\alpha-1}$ and $L^{\alpha-1}$ is downward-closed. 
This in particular implies $i_2 \in D^\alpha$. 

In particular, $Q$ contains no edge of $C_{\leftblock^\alpha-1}$, as all edges of $C_{\leftblock^\alpha-1}$ that are in the $t$-side of $Z$ lie on paths $P_i$, $i \notin L^{\alpha-1}$. 
Consequently, the subpath $Q'$ of $Q$ from $v$ to the tail of $e_2$ is completely contained in the $t$-side of $C_{\leftblock^\alpha-1}$ and the $s$-side of $C_{\rightblock^\alpha+\lambda}$. 
Concatenating $R$ with $Q'$ witnesses that the whole $Q'$ is in fact present in $G^\alpha$. 
Hence, $Q'$ with the edge $(u_{i_1},t)$ is a path from $v$ to $t$ in $G^\alpha-Z^\alpha$, as desired. 
\end{proof}

\begin{claim}\label{cl:project-Z}
Let $1 \leq \alpha \leq \xi$ be such that $\block^\alpha$ is an actual block and $\alpha$ is excellent. 
Then $Z^\alpha$ is a star $st$-cut in $G^\alpha$. 
Furthermore, $\corecut{Z^\alpha} = \corecut{Z} \cap Z^\alpha$. 
\end{claim}
\begin{proof}
This is an easy corollary of Claim~\ref{cl:excellent-reachability}. 

Since $Z$ is an $st$-cut, there is no vertex of $G$ that has both a path from $s$ to $v$ and a path from $v$ to $t$ in $G-Z$. 
Claim~\ref{cl:excellent-reachability} implies that no vertex of $G^\alpha$ admits both a path from $s$ to $v$ and from $v$ to $t$ in $G^\alpha-Z^\alpha$. Thus, $Z^\alpha$ is an $st$-cut.
The first point of Claim~\ref{cl:excellent-reachability} implies that from $Z$ being a star $st$-cut in $G$ follows that for every $(u,v) \in Z^\alpha$, $u$ is in the $s$-side of $Z^\alpha$ in $G^\alpha$
and $v$ is in the $t$-side of $Z^\alpha$ in $G^\alpha$. 
Finally, an arc $(u,v) \in Z^\alpha$ we have that $(u,v) \in \corecut{Z^\alpha}$ if and only
if there is path from $v$ to $t$ in $G^\alpha-Z^\alpha$, which by Claim~\ref{cl:excellent-reachability} is equivalent to an existence of a path from $v$ to $t$ in $G-Z$,
which is equivalent to $(u,v) \in \corecut{Z}$. 
\end{proof}

Claim~\ref{cl:project-Z} allows the following recursive step.
For every excellent $1 \leq \alpha \leq \xi$, recurse on $\inst^\alpha$, $\cutlb^\alpha$, and $\flow^\alpha$, obtaining a pair $(A^\alpha,\witnessflow^\alpha)$ in the randomized case
and a family $\mathcal{A}^\alpha$ in the deterministic case. 
By Claim~\ref{cl:progress}, all recursive calls either consider smaller value of $k$ or
consider properly boundaried instance with the same $k$, $\lambda$, and $H$, but fall into the $\ell \leq \ellthreshold$ case. 

There is now an important delicacy in the deterministic case. 
By duplicating some elements of some sets $\mathcal{A}^\alpha$, we can assume
that for fixed $\iota \in J$, all sets $\mathcal{A}^\alpha$ with $\iota(\alpha) = \iota$
have the same size, which we denote by $N_\iota$. Furthermore, enumerate
$\mathcal{A}_\alpha$ as $\{(A^\alpha_i,\witnessflow^\alpha_i)~|~1 \leq i \leq N_{\iota(\alpha)}\}$.
We iterate over all tuples of $(i_\iota)_{\iota \in J} \in \prod_{\iota \in J} \{1,2,\ldots,N_\iota\}$; for each such tuple $(i_\iota)_{\iota \in J}$ we denote
$(A^\alpha,\witnessflow^\alpha) = (A^\alpha_{i_{\iota(\alpha)}}, \witnessflow^\alpha_{i_{\iota(\alpha)}})$ for every excellent $\alpha$ and proceed 
as follows (i.e., in the same manner as in the randomized setting).

In the randomized setting,
we aim at $(A^\alpha,\witnessflow^\alpha)$ being compatible with $Z^\alpha$ for all excellent $\alpha$ where $\block^\alpha$ is an actual block. 
(Note that in the deterministic setting, this property is guaranteed
 for at least one choice of $(i_\iota)_{\iota \in J}$.)
In this case, for any such $\alpha$, we have $\cutlb^\alpha \leq \lambda_{G^\alpha+A^\alpha}(s,t) = |\corecutG{Z^\alpha}{G^\alpha+A^\alpha}| \leq |Z^\alpha| = k^\alpha$. 
For every $\iota \in J$, we randomly guess an integer $\cutlb_Z^\iota \leq \lambda_Z^\iota \leq k_Z^\iota$, aiming at $\lambda_Z^\iota = \lambda_{G^{\alpha_\iota} + A^{\alpha_\iota}}(s,t)$. 
The guess is correct with probability $2^{-\Oh(k \log k)}$,
and a straighforward branching replacement gives $2^{\Oh(k \log k)}$ subcases
  in the deterministic step.
We say that an excellent index $\alpha$ is \emph{superb} if $\lambda_{G^\alpha+A^\alpha}(s,t) = \lambda_Z^{\iota(\alpha)}$. 
Note that if our guess is correct, for every excellent $\alpha$, if $\block^\alpha$ is an actual block, then $\alpha$ is superb.

Return a set $A$ (or insert into the constructed set $\mathcal{A}$ in the deterministic setting)
consisting of:
\begin{enumerate}
\item for every superb $1 \leq \alpha \leq \xi$, the following edges:
\begin{enumerate}
\item 
  all edges of $A^\alpha$ with the following replacement:\label{edges:11}
\begin{itemize}
\item if $\leftblock^\alpha > 1$, replace every arc $(s,v) \in A^\alpha$ with arcs $(u,v)$ where $u$ ranges over all tails of edges of $C_{\leftblock^\alpha-1}$ on paths $P_i$ for $i \in D^\alpha$;
\item if $\rightblock^\alpha + \lambda \leq \ell$, replace every arc $(v,t) \in A^\alpha$ with arcs $(v,u)$ where $u$ ranges over all tails of edges of $C_{\rightblock^\alpha+\lambda}$ on paths $P_i$ for $i \in D^\alpha$;
\end{itemize}
\item for every $i \in [\lambda] \setminus D^\alpha$ an %
arc from the tail of the edge of $C_{\leftblock^\alpha-1} \cap E(P_i)$ (or $s$ if $\leftblock^\alpha=1$) to the tail of the edge of $C_{\rightblock^\alpha+\lambda} \cap E(P_i)$
(or $t$ if $\rightblock^\alpha+\lambda > \ell$);\label{edges:12}
\end{enumerate}
\item for every good but not superb $1 \leq \alpha \leq \xi$, the following edges:\label{edges:good}
  \begin{itemize}
  \item for every $i \in [\lambda]$, an edge from the tail of the edge of $C_{\leftblock^\alpha-1} \cap E(P_i)$ (or $s$ if $\leftblock^\alpha \leq 1$)
    to the tail of the edge of $C_{\rightblock^\alpha+\lambda} \cap E(P_i)$ (or $t$ if $\rightblock^\alpha+\lambda > \ell$). 
  \end{itemize}
\item for every $1 \leq a \leq \ell$ for which there is no good $\alpha$ such that the block $\block^\alpha = (\leftblock^\alpha < \ldots < \rightblock^\alpha)$ satisfies $\leftblock^\alpha \leq a \leq \rightblock^\alpha+1$:\label{edges:2}
\begin{itemize}
\item for every $i \in [\lambda]$, an edge from the tail of the edge of $C_a \cap E(P_i)$ to the tail of the edge of $C_{a+1} \cap E(P_i)$ (or $t$ if $a=\ell$);
\end{itemize}
\item for every $\Gamma$-milestone $a$, for every $\iota \in J$, for every $i,j \in L_Z^{\iota-1} \setminus L_Z^\iota$, 
an edge from the tail of the edge of $C_a \cap E(P_i)$ to the tail of the edge of $C_a \cap E(P_j)$. \label{edges:3}
\end{enumerate}

Accompany the set $A$ with a flow $\witnessflow$ of value $\sum_{\iota \in J} \lambda_Z^\iota$ 
constructed as follows.
Intuitively, for every $\iota \in [\eta]$, we push $\lambda_Z^\iota$ units of flow along the flow paths $P_i$, $i \in L_Z^{\iota-1} \setminus L_Z^\iota$.
We go from $s$ to $t$.
\begin{itemize}
\item For every superb $1 \leq \alpha \leq \xi$, denote $\iota := \iota(\alpha)$ and:
  \begin{itemize}
  \item
    push
    $\cutlb^\alpha = \cutlb_Z^\iota$ units of flow from $C_{\max(\leftblock^\alpha-1,1)}$ to $C_{\min(\rightblock^\alpha+\lambda,\ell)}$ using the flow $\witnessflow^\alpha$
  of value $\lambda_Z^\iota$ in $G^\alpha+A^\alpha$ (and arcs added in Point~\ref{edges:11}), using the arcs added in Point~\ref{edges:3} at milestones
  $\leftblock^\alpha-1$ (if $\leftblock^\alpha>1$) and $\rightblock^\alpha+\lambda$ (if $\rightblock^\alpha+\lambda \leq \ell$) to reshuffle the flow between flow paths $P_i$, $i \in L_Z^{\iota-1} \setminus L_Z^\iota$ in $C_{\leftblock^\alpha-1}$ and $C_{\rightblock^\alpha+\lambda}$, respectively;
  \item for every $\iota' \in J \setminus \{\iota\}$ push $\lambda_Z^{\iota'}$ units of flow 
    from $C_{\max(\leftblock^\alpha-1,1)}$ to $C_{\min(\rightblock^\alpha+\lambda,\ell)}$ using edges added in Point~\ref{edges:12}.
  \end{itemize}
\item For every good but not superb $1 \leq \alpha \leq \xi$,
  for every $\iota \in J$ push $\lambda_Z^\iota$ units of flow
  along the corresponding edges from $C_{\max(\leftblock^\alpha-1, 1)}$   to $C_{\min(\rightblock^\alpha+\lambda,\ell)}$ added in Point~\ref{edges:good}.
\item For every $1 \leq a \leq \ell$ that such that there is no good $\alpha$ with $\leftblock^\alpha-1 \leq a \leq \rightblock^\alpha+\lambda$, 
  for every $\iota \in J$ push $\lambda_Z^\iota$ units of flow
  along the corresponding edges added for $a$ and $i \in L_Z^{\iota-1} \setminus L_Z^{\iota}$ in Point~\ref{edges:2}.
\end{itemize}

Recall that we work under the assumption that all random guesses in this recursive call were correct for the cut $Z$ and, furthermore,
for every $\iota \in J$, the pair $(A^{\alpha_\iota},\witnessflow^{\alpha_\iota})$ is compatible with $Z^{\alpha_\iota}$ in $G^{\alpha_\iota}$.
  (In the deterministic setting, this is guaranteed for one of the choices of $(i_\iota)_{\iota \in J}$.)
The following two observations show correctness of the recursion above.
\begin{claim}\label{cl:Z-cut-size}
Let 
\[Z' := \bigcup_{\iota \in J} \corecutG{Z^{\alpha_\iota}}{G^{\alpha_\iota}+A^{\alpha_\iota}}.\]
Then $Z'$ is an $st$-cut in $G+A$. 
\end{claim}
\begin{proof}
By contradiction, let $Q$ be an $st$-path in $G+A-Z'$.

For every $\iota \in J$, if $\rightblock^{\alpha_\iota}+\lambda \leq \ell$, let $D_\iota$ be the set of all tails of edges of $C_{\rightblock^{\alpha_\iota}+\lambda}$.
Note that $D_\iota$ forms a vertex separator from $s$ to $t$ in $G+A$, that is, any path from $s$ to $t$ in $G+A$ needs to contain a vertex of $D_\iota$.

Assume that for some $\iota \in J$ for which $D_\iota$ is defined (i.e., $\rightblock^{\alpha_\iota} + \lambda \leq \ell$)
the path $Q$ contains a vertex $v \in D_\iota$ that is a tail of an edge of $C_{\rightblock^{\alpha_\iota}+\lambda}$ on path $P_i$ where $i \notin L_Z^\iota$. 
Let $\iota$, $v$, and $i$ be such that $v$ is the first such vertex on $Q$. For brevity, denote $\alpha = \alpha_\iota$.

As $L_Z^\iota$ is downward-closed, Claim~\ref{cl:no-escape} implies that no vertex on $Q$ before $v$ (except for $s$) lies on a path $P_{i'}$ for $i' \in L_Z^\iota$.
In particular, by the choice of $v$, there is no vertex of $D_\iota$ before $v$ on $Q$.

If $\leftblock^{\alpha} \leq 1$, then the prefix of $Q$ until $v$ is present in $G^\alpha + A^\alpha - Z^\alpha$, a contradiction to the assumption that $A^\alpha$ is compatible with $Z^\alpha$.
Otherwise, let $D$ be the set of tails of edges of $C_{\leftblock^{\alpha}-1}$ and observe that $Q$ needs to pass a vertex of $D$ before $v$. 
Let $w$ be the last such vertex and $j \in [\lambda]$ be such that $w$ is a tail of an edge of $C_{\leftblock^{\alpha}-1} \cap E(P_j)$.

Assume first that $j \notin L_Z^{\iota-1}$. Since $L_Z^0 = [\lambda]$ we have $\iota > 1$. 
The path $Q$ visits an edge of $C_{\rightblock^{\alpha_{\iota-1}} + \lambda}$ before $w$; let $v'$ be the tail of the last such edge before $w$ and assume this edge lies on path $P_{j'}$.
By the choice of $v$, $j' \in L_Z^{\iota-1}$. However, the subpath of $Q$ from $v'$ to $w$ contradicts Claim~\ref{cl:no-escape}.
Hence, we have $w \in L_Z^{\iota-1}$. Since $L_Z^{\iota-1}$ is downward-closed, Claim~\ref{cl:no-escape} implies that $i,j \in L_Z^{\iota-1} \setminus L_Z^\iota = D^{\alpha}$. 

Consider now the subpath of $Q$ from $w$ to $v$. By the choice of $w$, this subpath does not contain any vertex in the $s$-side of $C_{\leftblock^{\alpha_\iota}-1}$.
We have already established that $v$ is the first vertex on $Q$ that is a tail of an edge of $C_{\rightblock^{\alpha_\iota}+\lambda}$. 
We infer that the subpath of $Q$ from $w$ to $v$ is contained in the $t$-side of $C_{\leftblock^{\alpha_\iota}-1}$ and in the $s$-side of $C_{\rightblock^{\alpha_\iota}+\lambda}$ except for the last edge.
Hence, it projects to a path
from $s$ to $t$ in $G^\alpha+A^\alpha -Z^\alpha$, a contradiction to the assumption that $A^\alpha$ is compatible with $Z^\alpha$.
We infer that such a vertex $v$ does not exist. 

As $L_Z^\eta = \emptyset$, this is only possible if $\eta \in J$ and $\rightblock^{\alpha_\eta} + \lambda > \ell$. For brevity, denote $\alpha = \alpha_\eta$.
Since $\ell > \ellthreshold$, we have $\leftblock^{\alpha} > 1$. 
Let $D$ be the set of tails of edges of $C_{\leftblock^{\alpha}-1}$ and observe that $Q$ needs to pass a vertex of $D$.
Let $w$ be the last such vertex and assume $w$ is a tail of an edge of $C_{\leftblock^{\alpha}-1} \cap E(P_i)$ for some $i \in [\lambda]$.

Assume first that $i \notin L_Z^{\eta-1}$. Similarly as before, this implies $\eta > 1$ and, by Claim~\ref{cl:no-escape}, $Q$ visits an edge of $C_{\rightblock^{\alpha_{\eta-1}}+\lambda}$ on a path $P_{j'}$ with $j' \notin L_Z^{\eta-1}$.
But then $v$ and $\iota$ would have been defined, a contradiction. Hence, $i \in L_Z^{\eta-1}$.
By the choice of $w$, the subpath of $Q$ from $w$ to $t$ lies in the $t$-side of $C_{\leftblock^{\alpha}-1}$ and hence it projects to a path from $s$ to $t$ in $G^{\alpha}+A^{\alpha}-Z^{\alpha}$, 
a contradiction to the assumption that $A^{\alpha}$ is compatible with $Z^{\alpha}$. 
This finishes the proof of the claim.
\end{proof}
\begin{claim}\label{cl:Z-lifts}
$(A,\witnessflow)$ is compatible with $Z$ in $G$. 
\end{claim}
\begin{proof}
We first check if $A$ is compatible with $Z$ in $G$, that is, if no edge of $A$ goes from the $s$-side of $Z$ to the $t$-side of $Z$. 

For edges added in Point~\ref{edges:11}, consider two cases.
If $\block^\alpha$ is an actual block, it suffices to recall Claim~\ref{cl:excellent-reachability}: a vertex $v \in V(G^\alpha)$ is in the $s$-side of $Z^\alpha$
if and only if it is in the $s$-side of $Z$ in $G$ and,
furthermore for $i \in D^\alpha$ the endpoints of the edge of $E(P_i) \cap C_{\leftblock^\alpha-1}$ are in the $s$-side of $Z$ (if $\leftblock^\alpha > 1$)
and the endpoint of the edge of $E(P_i) \cap C_{\rightblock^\alpha+\lambda}$ (if $\rightblock^\alpha+\lambda \leq \ell$) are in the $t$-side of $Z$. 
If $\block^\alpha$ is not an actual block, we observe that all vertices of $V(G^\alpha) \setminus \{s,t\}$ lie either entirely in the $s$-side of $Z$ or entirely in the $t$-side of $Z$
due to the fact that we kept in $G^\alpha$ only vertices reachable from $s$ and from which one can reach $t$ in $G_0^\alpha$. 

For edges added in Point~\ref{edges:12}, we first observe
that for every $i \in [\lambda] \setminus D^\alpha$ it holds that $i \in L_{\leftblock^\alpha-1}$ if and only if $i \in L_{\rightblock^\alpha+\lambda}$:
if $\block^\alpha$ is an actual block it follows from the assumption that $L_{\leftblock^\alpha-1} = L^{\alpha-1}$ and $L_{\rightblock^\alpha+\lambda} = L^\alpha$
and otherwise we have $L_{\leftblock^\alpha-1} = L_{\rightblock^\alpha+\lambda}$. 
Hence, as both $\leftblock^\alpha-1$ (if $\leftblock^\alpha > 1$) and $\rightblock^\alpha+\lambda$ (if $\rightblock^\alpha+\lambda \leq \ell$) are milestones, 
we have that any edge added in Point~\ref{edges:12} lie on the same side of $Z$. 

For edges added in Point~\ref{edges:good}, note that we have $L_{\leftblock^\alpha-1} = L^{\alpha-1} = L^\alpha = L^{\rightblock^\alpha+\lambda}$, so an edge added for $i \in L^\alpha$
connects two vertices in the $s$-side of $Z$ and for $i \notin L^\alpha$ connects two vertices in the $t$-side of $Z$.

For edges added in Point~\ref{edges:2}, note that any such $a$ is untouched and hence for every $i \in [\lambda]$, the entire subpath of $P_i$ from the edge of $C_a$
to the edge of $C_{a+1}$ (or $t$ if $a=\ell$) lies in the same side of $Z$. Hence, again for any edge added in Point~\ref{edges:2}, both its endpoints are in the same side of $Z$.

Finally, for edges added in Point~\ref{edges:3}, the claim follows from the fact that every $\Gamma$-milestone $a$ is a milestone, and hence $a$ is untouched and $L_a \in \{L_Z^0, \ldots, L_Z^\eta\}$. 
Thus, both endpoints of an edge added in Point~\ref{edges:3} lies in the same side of $Z$. 

This concludes the proof that $A$ is compatible with $Z$ in $G$. 
In particular, $Z$ is a star $st$-cut in $G+A$. 

Claim~\ref{cl:Z-cut-size} asserts that $|Z'|$ is an $st$-cut in $G+A$ of cardinality $|\witnessflow|$. Hence, $\witnessflow$ is an $st$-maxflow in $G+A$ and $Z'$ is an $st$-mincut in $G+A$.
As $Z' \subseteq Z$, we have that $Z$ is a star $st$-cut and $Z' = \corecutG{Z}{G+A}$. 

It remains to check  that $\witnessflow$ is a witnessing flow for $Z$ in $G+A$.
Note that the only edges of $\witnessflow$ that are not in $A$ are edges inside subinstances $G^\alpha$ from recursive calls for superb $\alpha$.
If $\block^\alpha$ is an actual block, then $\witnessflow$ does not use any edge of $Z \setminus Z'$ in $G^\alpha$ as $\witnessflow^\alpha$ is a witnessing flow for $Z^\alpha$ in $G^\alpha+A^\alpha$.
If $\block^\alpha$ is not an actual block, there is no edge of $Z$ in $G^\alpha$.
This finishes the proof of the claim.
\end{proof}

\subsection{Probability analysis}\label{ss:prob-analysis}

Here we focus on the randomized setting only.

Fix a star $st$-cut $Z$ of $|Z| \leq k$.
We want to lower bound the probability that the recursive algorithm finds a pair $(A,\witnessflow)$ compatible with $Z$.
We say that the algorithm is \emph{correct} at some random step if it correctly guesses the properties of $Z$.

In the preprocessing phase, the algorithm either directly returns a desired pair (without any random choice), recurses on an instance with 
smaller value of $2k-\lambda$ (being correct with probability $\Omega(k^{-1})$, or recurses on an instance with the same value of $k$ and $\lambda$
and, being properly boundaried (with probability $\Omega(1)$). 

In the base case $|E(H)| = |V(H)|$, we either return a correct outcome or recurse on an instance with larger $\lambda$ (making correct choices with probability $\Omega(\lambda^{-1}) = \Omega(k^{-1})$).

In the case of $\ell \leq \ellthreshold$, the algorithm makes correct guesses about sets $B_\leftarrow$ and $B_\rightarrow$ with probability at least $2^{-\Oh(k \log k)}$. 
If $\lambda_{G+A_0}(s,t) > \lambda_{G}(s,t)$ at this point, the algorithm recurses on an instance with smaller value of $2k-\lambda$. 
If a vertex of $V(C)$ is in the $s$-side of $Z$, where $C$ is the $st$-mincut of $G+A_0$ closest to $t$, 
the algorithm makes correct guesses with probability $\Omega(k^{-1})$ and again recurses on an instance with a smaller value of $2k-\lambda$.
If $V(C)$ is contained in the $t$-side of $Z$, the algorithm guess so with probability $\Omega(1)$ and recurses on an instance with the same $k$ and $\lambda$, but smaller $|E(H)|$.
Note that $|E(H)|$ can decrease (while keeping $k$ and $\lambda$ the same) less than $k^2$ times.

In the case of $\ell > \ellthreshold$, recall the algorithm makes correct guesses with probability $2^{-\Oh(k \log k)}$. 
Recall also that the algorithm returns a correct output if it makes correct guesses and the following holds: in every $\Gamma$-block $\block^\alpha$ that is an actual block and $\alpha$ is superb, 
the returned pair $(A^\alpha,\witnessflow^\alpha)$ from the recursive call is compatible with $Z^\alpha$ in $G^\alpha$.
By Claim~\ref{cl:progress}, there is either one such recursive call (if $L_Z^{\iota-1} \setminus L_Z^\iota = [\lambda]$ for some $1 \leq \iota \leq \eta$)
with the same $k$, $\lambda$, $H$, but falling into case $\ell \leq \ellthreshold$, 
or a number of such calls with strictly smaller values $2k^\alpha-|\flow^\alpha|$ summing up to at most $2k-|\flow|$. 

Hence, the depth of the recursion is $\Oh(k^3)$, as $2k-\lambda$ can decrease only $2k$ times and, between these decreases, $|E(H)|$ can decrease less than $k^2$ times.
Furthermore, in the entire recursion tree we care only about being correct on $\Oh(k^3)$ recursive calls, as even if in the last case there are multiple subcalls $(G^\alpha,\flow^\alpha)$
with $Z^\alpha \neq \emptyset$, the values of $2k^\alpha-|\flow^\alpha|$ for these calls sum up to at most $2k-|\flow|$. 
At each recursive call, we make correct guesses with probability $2^{-\Oh(k \log k)}$. 
Hence, the overall success probability is $2^{-\Oh(k^4 \log k)}$, as desired.

\subsection{Branching analysis}\label{ss:det-analysis}

Here we focus on the deterministic setting only.

We want to upper bound the number of elements output by the algorithm.
In the preprocessing phase, the algorithm either directly returns, or branches
into $\Oh(k)$ instance with 
smaller value of $2k-\lambda$ and a single instance with the same value of $k$ and $\lambda$
and, being properly boundaried.

In the base case $|E(H)| = |V(H)|$, we insert one element into $\mathcal{A}$
in one branch and furthermore branch into $\lambda \leq k$ instances with larger $\lambda$.

In the case of $\ell \leq \ellthreshold$, the algorithm branches into $2^{\Oh(k \log k)}$
choices for the sets $B_\leftarrow$ and $B_\rightarrow$.
In some branches $\lambda_{G+A_0}(s,t) > \lambda_{G}(s,t)$ or an edge of the sought cut $Z$ is found
and the algorithm recurses on an instance with smaller value of $2k-\lambda$.
In one branch, the algorithm recurses
on an instance with the same values of $k$ and $\lambda$, but smaller number of edges of $H$.
Then, the algorithm branches in the $s$-/$t$-assignment of the vertices of $V(C)$. 
In all but one branches where a vertex of $V(C)$ is in the $s$-side of $Z$, where $C$ is the $st$-mincut of $G+A_0$ closest to $t$, 
the algorithm again recurses on an instance with a smaller value of $2k-\lambda$.
In the final branch when $V(C)$ is contained in the $t$-side of $Z$, 
the algorithm recurses on an instance with the same $k$ and $\lambda$, but smaller $|E(H)|$.
Note that $|E(H)|$ can decrease (while keeping $k$ and $\lambda$ the same) less than $k^2$ times.

In the case of $\ell > \ellthreshold$, recall the algorithm takes
$2^{\Oh(k \log k)} \cdot \Oh(\log^2 n)$ branches
and, furthermore, iterates over all choices of $(i_\iota)_{\iota \in J}$ where
$1 \leq i_\iota \leq N_\iota$ for every $\iota \in J$. 
By Claim~\ref{cl:progress},
we have either $|J| = 1$ and the recursive calls for excellent indices $\alpha$
has the same $k$, $\lambda$, $H$, but falling into case $\ell \leq \ellthreshold$, 
or $|J| > 2$
  and 
\begin{equation}\label{eq:det}
  \sum_{\iota \in J} 2k_Z^\iota-|L_Z^{\iota-1} \setminus L_Z^\iota| \leq 2k-\lambda.
\end{equation}
  Note that $k^\alpha = k_Z^\iota$ and $\lambda^\alpha = |D^\alpha| = |L_Z^{\iota-1} \setminus L_Z^\iota|$ for excellent $\alpha$ with $\iota(\alpha) = \iota$.

Hence, the depth of the recursion is $h = \Oh(k^3)$, as $2k-\lambda$ can decrease only $2k$ times and, between these decreases, $|E(H)|$ can decrease less than $k^2$ times.
Consequently, by standard bottom-up induction using~\eqref{eq:det},
  the set $\mathcal{A}$ returned at depth $d$
of the recursion has size bounded by $2^{\Oh((2k-\lambda) (h-d) \log k)} \cdot \log^{2(h-d)} n$.
The desired bound
on the running time and the final size of the returned family $\mathcal{A}$ follows.

\section{Applications}
\subsection{\textsc{Weighted Bundled Cut} with pairwise linked deletable edges}

Recall that
an instance of \textsc{Bundled Cut} consists of a directed multigraph $G$, vertices $s,t \in V(G)$, a nonnegative integer $k$, and a family $\bundles$ of pairwise disjoint subsets of $E(G)$.
An element $B \in \bundles$ is called a \emph{bundle}. An edge that is part of a bundle is
\emph{soft}, otherwise it is \emph{crisp}.
A \emph{cut} in a \textsc{Bundled Cut} instance $\inst = (G,s,t,k,\bundles)$ is an $st$-cut $Z$ that does not contain any crisp edge, that is, $Z \subseteq \bigcup \bundles$. 
A cut $Z$ \emph{touches} a bundle $B \in \bundles$ if $Z \cap B \neq \emptyset$. 
The \emph{cost} of a cut $Z$ is the number of bundles it touches.
A cut $Z$ is a \emph{solution} if its cost is at most $k$. 
The \textsc{Bundled Cut} problem asks if there exists a solution to the input instance. 

An instance of \textsc{Weighted Bundled Cut} consists of a \textsc{Bundled Cut} instance $\inst = (G,s,t,k,\bundles)$ and additionally a weight function $\weight : \bundles \to \mathbb{Z}_+$ and an integer $W \in \mathbb{Z}_+$. 
The \emph{weight} of a cut $Z$ in $\inst$ is the total weight of all bundles it touches. 
A solution $Z$ to $\inst$ is a solution to the \textsc{Weighted Bundled Cut} instance $(\inst,\weight,W)$ if additionally the weight of $Z$ is at most $W$. 
The \textsc{Weighted Bundled Cut} problem asks if there is a solution to the input instance.
Note that any \textsc{Bundled Cut} instance can be treated as a \textsc{Weighted Bundled Cut} instance by setting $\weight$ uniformly equal $1$ and $W=k$. 

An edge $e \in E(G)$ is \emph{deletable} if it is soft and there is no parallel arc to $e$
that is crisp, and \emph{undeletable} otherwise.
An instance $(G,s,t,k,\bundles)$ of \textsc{Bundled Cut}
has \emph{pairwise linked deletable edges} if for every
$B \in \bundles$ and every two deletable edges $e_1,e_2 \in B$, there exists a path in $G$
from an endpoint of $e_1$ to an endpoint of $e_2$ that uses only edges of $B$
and undeletable edges. 

In this subsection we prove the following theorem.
\begin{theorem}\label{thm:bundled-cut}
\textsc{Weighted Bundled Cut}, restricted to instances with pairwise linked deletable
edges, can be solved in time $2^{\Oh(k^4 d^4 \log(kd))} n^{\Oh(1)}$, 
where $d$ is the maximum number of deletable edges in a single bundle.
\end{theorem}
\begin{proof}
We can restrict ourselves to search for cuts $Z$ that contain
only deletable edges and $|Z| \leq kd$.
We invoke the deterministic version of the directed flow-augmentation for $(G,s,t,kd)$. 
By iterating over the resulting sets $A$,
we can end up with an instance where $Z$ is actually an $st$-mincut. 
Thus, by somehow abusing the notation, 
  we assume that $Z$ is already an $st$-mincut in $\inst$.
Let $\witnessflow = \{P_1,P_2,\ldots,P_\lambda\}$ be any maximum flow in $G$
for $\lambda = |Z| \leq kd$.

Note that any $st$-mincut $Z$ contains exactly one deletable edge on 
every path $P_i$. Let $f_i$ be the unique edge of $E(P_i) \cap Z$ for a fixed
hypothethical solution $Z$.
We perform the following branching and color coding steps.

First, branch by guessing the number $\kappa \leq k$ of bundles violated by $Z$;
let $B_1,B_2,\ldots,B_\kappa$ be the violated bundles. 
For every $1 \leq i \leq \lambda$ we guess the index $\alpha(i) \in [\kappa]$
such that $f_i$ lies in $B_{\alpha(i)}$; this gives $2^{\Oh(k \log (kd))}$ subcases.
For every $B \in \bundles$, we guess $\gamma(B) \in [\kappa]$, aiming
at $\gamma(B_j) = j$ for $1 \leq j \leq \kappa$. 
Using Theorem~\ref{thm:color-coding}, this can be done by branching into $2^{\Oh(k \log k)} \cdot \Oh(\log n)$ cases. 
For every $B \in \bundles$ and every $1 \leq i \leq \lambda$
such that $\alpha(i) = \gamma(B)$, we guess a deletable edge $e(B,i) \in B$, aiming
at $f_i = e(B_{\alpha(i)},i)$.
Again using Theorem~\ref{thm:color-coding}, this can be done by branching into $2^{\Oh(k^2 d \log(kd))} \cdot \Oh(\log n)$ cases.
For every $B \in \bundles$, we make every edge of $B$ that is not an edge
$e(B,i)$ for some $i \in [\lambda]$ undeletable (i.e., we add a crisp copy if there is none).

We make a sanity check: we expect that every edge $e(B,i)$ actually lies on the path $P_i$
and is a deletable edge; otherwise we delete the bundle $B$ (making all its edges crisp).

Finally, we make the following guessing step. For every $1 \leq j \leq \kappa$
and every two distinct $1 \leq i_1,i_2 \leq \lambda$ such that $\alpha(i_1) = \alpha(i_2) = j$, we guess the relation between $f_{i_1}$ and $f_{i_2}$ in the bundle $B_j$
according to the definition of having pairwise linked deletable arcs.
That is, note that one of the following holds:
\begin{enumerate}
\item $f_{i_1}$ or $f_{i_2}$ has a tail in $s$;
\item $f_{i_1}$ or $f_{i_2}$ has a head in $t$;
\item there is a path from an endpoint of $f_{i_1}$ to an endpoint of $f_{i_2}$ that does not use arcs of other bundles;
\item there is a path from an endpoint of $f_{i_2}$ to an endpoint of $f_{i_1}$ that does not use arcs of other bundles.
\end{enumerate}
We guess the first case that applies and, additionally, if one of the first two cases applies,
we guess which of the four subcases applies (whether $f_{i_1}$ or $f_{i_2}$ and
whether tail in $s$ or head in $t$). 
This guessing step results in is $2^{\Oh(k\lambda^2)} = 2^{\Oh(k^3d^2)}$ subcases.

We delete all bundles from $\bundles$ that do not comply with the guess above. Note that
for every $1 \leq j \leq \kappa$ for which one of the two first options is guessed,
at most one bundle $B$ with $\gamma(B) = j$ remains, as there is at most one deletable
arc with tail in $s$ (head in $t$) on a single flow path.

\begin{figure}[tb]
\begin{center}
\begin{tikzpicture}[scale=1.25]
\draw (-5,0) node {$s$};
\draw (5,0) node {$t$};
\foreach \y in {1,2,3} {
  \draw[->] (-4.5, 1-\y*0.5) -- (4.5, 1-\y*0.5) node[pos=0.1,above] {\tiny $P_\y$};
}
\draw[very thick,black!50!green,dashed,->] (0.5, 0.5) -- (1.0, 0.5) -- (0, 0) -- (0.5, 0);
\draw[very thick,black!50!magenta,dashed,->] (-0.5, 0.5) -- (0, 0.5) -- (1, 0) -- (1.5, 0);
\end{tikzpicture}
\caption{The crucial filtering step in the proof of Theorem~\ref{thm:bundled-cut}
  in the special case of 3-\textsc{Chain SAT} problem where the bundles are paths consisting of at most three edges.
  The algorithm already guessed that there is a violated bundle that has its first edge on $P_1$ and its last edge on $P_2$. 
  The green candidate bundle cannot be the one violated by the solution, as the violet candidate
  provides a bypass from a vertex on $P_1$ before the green bundle to a vertex on $P_2$ after the green bundle.}\label{fig:bundled-cut}
\end{center}
\end{figure}
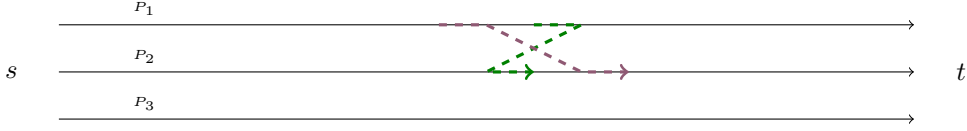

We now make the following (crucial) filtering step.
Iterate over all $1 \leq j \leq \kappa$ and indices $1 \leq i_1, i_2 \leq \lambda$ such that $\alpha(i_1) = \alpha(i_2) = j$ and $i_1 \neq i_2$. 
Consider two bundles $B,B'$ with $\gamma(B) = \gamma(B') = j$ such that 
$e(B,i_1)$ is before $e(B',i_1)$ on $P_{i_1}$ but $e(B,i_2)$ is after $e(B',i_2)$ on $P_{i_2}$.
If the guessed relation for $(j,i_1,i_2)$ is of the third type, then
it cannot hold that $B' = B_j$, as otherwise the said connection for the bundle $B$,
together with a path from $s$ to $e(B,i_1)$ along $P_{i_1}$ and a path from $e(B,i_2)$
to $t$ along $P_{i_2}$ form an $st$-path avoiding $Z$ (see Figure~\ref{fig:bundled-cut}).
Thus, we can delete $B'$ from $\bundles$.
Symmetrically, if the guessed relation for $(j,i_1,i_2)$ is of the fourth
type, then it cannot hold that $B = B_j$, so we can delete $B$ from $\bundles$.

Once we perform the above filtering step exhaustively, for every $1 \leq j \leq \kappa$ the
bundles $B$ with $\gamma(B) = j$ can be enumerated as $B_{j,1}, \ldots, B_{j,n_j}$ such that for every $1 \leq \xi < \zeta \leq n_j$
and for every $i \in [\lambda]$, if $\alpha(i) = j$, then $e(B_{j,\xi},i)$ is before $e(B_{j,\zeta},i)$ on $P_i$. 
Let $a_j^Z$ be such that $B_j = B_{j,a_j^Z}$. 

Note that at this point the deletable edges are exactly the edges $e(B,i)$ for some $i \in [\lambda]$ and $B \in \bundles$.
If all guesses are successful, $Z$ is still a solution and all edges of $Z$ are deletable.

We now construct an auxiliary weighted directed graph $H$ as follows.
Start with $H$ consisting of two vertices $s$ and $t$.
For every $1 \leq j \leq \kappa$, add a path $P^H_j$ from $s$ to $t$ with $n_j$ edges; denote the $a$-th edge as $e_{j,a}$ and set its weight as $\weight(e_{j,a}) = \weight(B_{j,a})$. 
Furthermore, for every $1 \leq i_1,i_2 \leq \lambda$, denote $j_1 = \alpha(i_1)$, $j_2 = \alpha(i_2)$, 
for every $1 \leq a_1 \leq n_{j_1}$ and $1 \leq a_2 \leq n_{j_2}$,
for every endpoint $u_1$ of $e(B_{j_1,a_1}, i_1)$, 
for every endpoint $u_2$ of $e(B_{j_2,a_2}, i_2)$,
if $G$ contains a path from $u_1$ to $u_2$ consisting only of crisp edges,
then add to $H$ an edge of weight $+\infty$
from the corresponding endpoint of $e_{j_1,a_1}$ (i.e., tail if and only if $u_1$ is a tail of $e(B_{j_1,a_1},i_1)$)
to the corresponding endpoint of $e_{j_2,a_2}$ (i.e., tail if and only if $u_2$ is a tail of $e(B_{j_2,a_2},i_2)$).

Observe that $Z' := \{e_{j,a_j^Z}~|~1 \leq j \leq \kappa\}$ is an $st$-cut in $H$. Indeed, if $H$ would contain an arc $(v,u)$ with $v$ before $e_{j,a_{j}^Z}$ and $u$ after $e_{j',a_{j'}^Z}$ for some $j,j' \in [\kappa]$,
then this arc was added to $H$ because of some path between the corresponding endpoints in $G$ and such a path would lead from the $s$-side to $t$-side of $Z$. 
Also, in the other direction, observe that if $Y' = \{e_{j,a_j}~|~1 \leq j \leq \kappa\}$ is an $st$-mincut in $H$, then 
\[ Y = \left\{ e(B_{\alpha(i),a_{\alpha(i)}},i)~|~i \in [\lambda]\right\} \]
is a solution to $\inst$ of the same weight. 

Hence, it suffices to find in $H$ an $st$-cut of cardinality $\kappa$ and minimum possible weight.
Since $\kappa \leq \lambda_H(s,t)$,
this can be done in polynomial time by a reduction to the task of finding an $st$-cut of minimum capacity: we set the capacity of an edge $e$ as $\weight(e) + 1 + \sum_{f \in E(H)} \weight(f)$. 

This finishes the proof of Theorem~\ref{thm:bundled-cut}.
\end{proof}

\subsection{Weighted Directed Feedback Vertex Set}\label{ss:wdfvs}

First, by standard reductions between the edge-deletion versions and vertex-deletion versions,
for Corollary~\ref{cor:wdfvs} we can actually solve the edge-deletion version.
\textsc{Weighted Directed Feedback Arc Set} (\textsc{Weighted DFAS}).
By standard approach, \textsc{Weighted DFAS} can be solved using a subroutine for \textsc{Weighted Skew Multicut}.
Here, we are given a directed graph $G$, a tuple $(s_i,t_i)_{i=1}^b$ of terminal pairs, a weight function $\weight: E(G) \to \mathbb{Z}_+$, and integers $k,W$.
The goal is to find a set $Z \subseteq E(G)$ of cardinality at most $k$, weight at most $W$, and such that there is no path from $s_i$ to $t_j$ in $G-Z$ for any $1 \leq i \leq j \leq b$. 
We observe the following reduction.
\begin{lemma}\label{lem:wdfvs}
Given a \textsc{Weighted Skew Multicut} instance $\inst = (G,(s_i,t_i)_{i=1}^b,\weight,k,W)$, one can in polynomial time construct an equivalent
\textsc{Weighted Bundled Cut} instance $\inst' = (G',\bundles,\weight,k,W)$ with the same $k$ and $W$, where each bundle has at most $b$ deletable edges
and the instance has pairwise linked deletable edges.
\end{lemma}
\begin{proof}
To construct the graph $G'$, we start with $b$ disjoint copies $G^1, \ldots, G^b$ of the graph $G$. 
By $v^i$, $e^i$, etc., we denote the copy of vertex $v$ or edge $e$ in the copy $G^i$. 
We set $B_e = \{e^i~|~i \in [b]\}$ for $e \in E(G)$ and $\bundles = \{B_e~|~e \in E(G)\}$, that is, all $b$ copies of one edge of $G$ form a bundle. 
We set weights of bundles as $\weight'(B_e) = \weight(e)$. 
There will be no more bundles, so all arcs introduced later to $G'$ are crisp. 

\begin{figure}[tb]
\begin{center}
\begin{tikzpicture}[]
\tikzstyle{vertex}=[circle,fill=black,minimum size=0.2cm,inner sep=0pt]
\node[vertex] (s) at (-1, -3) {};
\node[vertex] (t) at (7, -3) {};
\draw[left] (s) node {$s$};
\draw[right] (t) node {$t$};
\foreach \n in {1,2,3,4,5} {
  \begin{scope}[shift={(0,-\n)}]
  \draw[rounded corners=5pt] (-0.2, -0.2) rectangle (6.2, 0.2);
  \draw (4, 0) node {$G^\n$};
    \node[vertex] (s\n) at (-0.2+\n*0.2, 0) {};
    \node[vertex] (t\n) at (5+\n*0.2, 0) {};
    \draw[right] (s\n) node {$s_\n^\n$};
    \draw[left] (t\n) node {$t_\n^\n$};
    \draw[->] (s) -- (s\n);
    \draw[->] (t\n) -- (t);
    \node[vertex] (x\n) at (2, 0) {};
    \node[vertex] (y\n) at (3, 0) {};
    \draw[->,very thick,red] (x\n) -- (y\n);
  \end{scope}
}
\foreach \a/\b in {1/2,2/3,3/4,4/5} {
  \draw[->] (x\a) -- (x\b);
  \draw[->] (y\a) -- (y\b);
}
\foreach \a/\b in {1/3,2/4,3/5} {
  \draw[->] (x\a) edge[out=260,in=100] (x\b);
  \draw[->] (y\a) edge[out=280,in=80] (y\b);
}
\foreach \a/\b in {1/4,2/5} {
  \draw[->] (x\a) edge[out=250,in=110] (x\b);
  \draw[->] (y\a) edge[out=290,in=70] (y\b);
}
\foreach \a/\b in {1/5} {
  \draw[->] (x\a) edge[out=240,in=120] (x\b);
  \draw[->] (y\a) edge[out=300,in=60] (y\b);
}
\end{tikzpicture}
\caption{The construction used in Lemma~\ref{lem:wdfvs}.
The copies of a single edge of $G$, marked in red, form a single bundle.}\label{fig:wdfvs}
\end{center}
\end{figure}
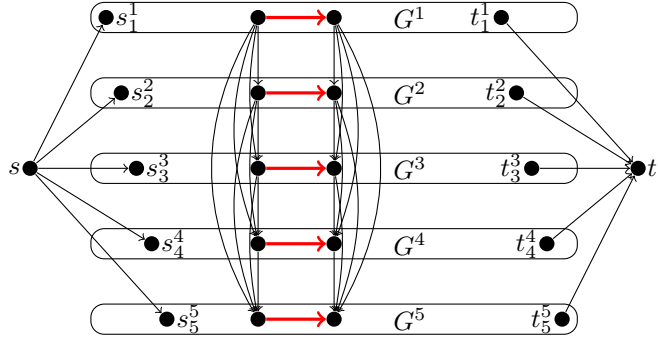

For every $1 \leq i < j \leq b$ and $v \in V(G)$, we add to $G'$ an arc $(v^i,v^j)$. Note that these arcs make the instance satisfy the pairwise linked deletable edges property.
Furthermore, we add to $G'$ vertices $s$ and $t$ and arcs $(s,s_i^i)$ and $(t_i^i,t)$ for every $i \in [b]$. This finishes the description of the instance $\inst' = (G',\bundles,\weight',k,W)$. See Figure~\ref{fig:wdfvs} for an illustration. 

It is straightforward to observe that if $Z$ is a solution to the instance $\inst$, then $\bigcup_{e \in Z} B_e$ is a solution to $\inst'$ of the same weight and touching $|Z|$ bundles.
In the other direction, note that if $Z'$ is a solution to $\inst'$ then $Z = \{e \in E(G)~|~Z' \cap B_e \neq \emptyset\}$ is a solution to $\inst$.
\end{proof}

We deduce the following.
\begin{theorem}
\textsc{Weighted DFAS} and \textsc{Weighted DFVS} can be solved
in time $2^{\Oh(k^8 \log k)} n^{\Oh(1)}$, that is, are FPT when parameterized by $k$.
\end{theorem}

\bibliographystyle{alpha}
\bibliography{../references.bib}
\end{document}